\documentclass[12pt]{article}

\usepackage{graphicx}
\usepackage{amsmath}
\usepackage{amssymb}
\usepackage{bm}
\usepackage{amsthm}
\usepackage{tikz}
\usetikzlibrary{decorations.markings,decorations.pathreplacing,calc}
\usepackage{algpseudocode}
\usepackage{algorithm}
\usepackage{thm-restate}
\usepackage{verbatim}
\usepackage{fancyvrb}
\usepackage{subfig}
\usepackage{url}
\usepackage{color}
\usepackage{cancel}
\usepackage{hyperref}
\usepackage{authblk}
\usepackage{listings}

\input{Qcircuit}

\newcommand{\be}{\begin{equation}}
\newcommand{\ee}{\end{equation}}
\newcommand{\ba}{\begin{array}}
\newcommand{\ea}{\end{array}}
\newcommand{\bea}{\begin{eqnarray}}
\newcommand{\eea}{\end{eqnarray}}

\newcommand{\ra}{\rangle}
\newcommand{\la}{\langle}

\newcommand{\irange}[2]{[{#1}..{#2}]}

\newcommand{\calB}{{\cal B}}

\newcommand{\calL}{{\cal L }}

\newcommand{\calF}{{\cal F }}

\newcommand{\calC}{{\cal C }}
\newcommand{\calS}{{\cal S }}

\newcommand{\calP}{{\cal P }}
\newcommand{\calM}{{\cal M }}

\newcommand{\FF}{\mathbb{F}}

\newcommand{\ZZ}{\mathbb{Z}}

\newcommand{\Gate}[1]{\textsc{#1}}
\newcommand{\hgate}{\Gate{h}}
\newcommand{\zgate}{\Gate{z}}
\newcommand{\ygate}{\Gate{y}}
\newcommand{\xgate}{\Gate{x}}
\newcommand{\czgate}{\Gate{cz}}

\newcommand{\pgate}{\Gate{p}}
\newcommand{\idgate}{\Gate{i}}
\newcommand{\notgate}{\Gate{not}}
\newcommand{\cnotgate}{\Gate{cnot}}
\newcommand{\cxgate}{\Gate{cnot}}
\newcommand{\swapgate}{\Gate{swap}}

\newcommand{\myAND}{{\bf and} }

\newtheorem{proposition}{Proposition}
\newtheorem{lemma}{Lemma}

\newtheorem{conj}{Conjecture}
\newtheorem{theorem}{Theorem}
\newtheorem{example}{Example}

\newcommand{\eq}[1]{Eq.~(\ref{eq:#1})}
\renewcommand{\sec}[1]{\hyperref[sec:#1]{Section~\ref*{sec:#1}}}
\newcommand{\ssec}[1]{\hyperref[ssec:#1]{Subsection~\ref*{ssec:#1}}}
\newcommand{\fig}[1]{\hyperref[fig:#1]{Fig.~\ref*{fig:#1}}}
\newcommand{\tab}[1]{\hyperref[tab:#1]{Table~\ref*{tab:#1}}}
\newcommand{\lem}[1]{\hyperref[lem:#1]{Lemma~\ref*{lem:#1}}}
\newcommand{\prop}[1]{\hyperref[prop:#1]{Proposition~\ref*{prop:#1}}}
\newcommand{\thm}[1]{\hyperref[thm:#1]{Theorem~\ref*{thm:#1}}}

\begin{document}

\title{Hadamard-free circuits expose the structure of the Clifford group}

\author{Sergey Bravyi and Dmitri Maslov}
\affil{IBM T. J. Watson Research Center, Yorktown Heights, NY 10598, USA}

\maketitle

\abstract{The Clifford group plays a central role in quantum randomized benchmarking, quantum tomography, and error correction protocols.  Here we study the structural properties of this group.  We show that any Clifford operator can be uniquely written in the canonical form $F_1HSF_2$, where $H$ is a layer of Hadamard gates, $S$ is a permutation of qubits, and $F_i$ are parameterized Hadamard-free circuits chosen from suitable subgroups of the Clifford group.  Our canonical form provides a one-to-one correspondence between Clifford operators and layered quantum circuits.  We report a polynomial-time algorithm for computing the canonical form.  We employ this canonical form to generate a random uniformly distributed $n$-qubit Clifford operator in runtime $O(n^2)$.  The number of random bits consumed by the algorithm matches the information-theoretic lower bound.  A surprising connection is highlighted between random uniform Clifford operators and the  Mallows distribution on the symmetric group.  The variants of the canonical form, one with a short Hadamard-free part and one allowing a circuit depth $9n$ implementation of arbitrary Clifford unitaries in the Linear Nearest Neighbor architecture are also discussed.  Finally, we study computational quantum advantage where a classical reversible linear circuit can be implemented more efficiently using Clifford gates, and show an explicit example where such an advantage takes place.
}

\section{Introduction}

Clifford circuits can be defined as quantum computations by the circuits with Phase ($\pgate$), Hadamard ($\hgate$), and $\cnotgate$ gates, applied to a computational basis state, such as $\ket{00...0}$.  While not universal for the purpose of quantum (or classical) computation and admitting an efficient classical simulation \cite{aaronson2004improved}, Clifford circuits play an important role in quantum computing: they underlie quantum error correction \cite{nielsen2010quantum} and serve as distillation subroutines in the fault-tolerant computation \cite{bravyi2005universal}.  Optimization of Clifford circuits is necessary within the scope of fault-tolerant computations, since it can happen that the Clifford overhead dominates the cost of an implementation \cite{maslov2016optimal}, despite the cost of the non-Clifford gates being higher than that of the Clifford gates. Of particular interest are algorithms for generation and compiling of random uniformly distributed Clifford operators~\cite{koenig2014efficiently}. The latter are known to form a unitary $3$-design~\cite{zhu2017multiqubit,webb2015clifford} and  play a central role in the characterization of noise in quantum computers via randomized benchmarking \cite{emerson2005scalable, knill2008randomized,magesan2011scalable}, efficient tomography of multi-qubit states~\cite{aaronson2018shadow,huang2020predicting}, compressed classical description of quantum states~\cite{gosset2018compressed}, randomized quantum code construction~\cite{bennett1996mixed}, and quantum data hiding~\cite{divincenzo2002quantum}.  Finally, Clifford circuits give rise to interesting toy models of topological quantum order and entanglement renormalization in quantum many-body systems~\cite{aguado2008entanglement,haah2014bifurcation}.

Obtaining an efficient circuit implementation of a Clifford group element is a problem that was studied well in the literature.  Aaronson and Gottesman~\cite{aaronson2004improved} showed that any Clifford operator admits an  11-stage decomposition of the form  -H-CX-P-CX-P-CX-H-P-CX-P-, where -H-, -P-, and -CX- represent circuit stages with $\hgate$, $\pgate$, and $\cnotgate$ gates, correspondingly.  This decomposition is asymptotically optimal (optimal up to a constant factor) in terms of the number of degrees of freedom as well as the number of gates~\cite{patel2008optimal}.  Roetteler and one of the authors~\cite{maslov2018shorter} leveraged Hadamard-free circuits (those including only $\pgate$, $\czgate$, $\cxgate$ and Pauli gates) to expose more structure of the Clifford group.  By employing Bruhat decomposition of the underlying symplectic group, they showed that any Clifford operator admits a 7-stage decomposition -CX-CZ-P-H-P-CZ-CX-, where -CZ- corresponds to a layer of $\czgate$ gates. The number of degrees of freedom in \cite{maslov2018shorter} is asymptotically tight (saturates the information-theoretic lower bound within a constant factor that asymptotically approaches $1$).
 
The main result of this paper is a canonical form of Clifford circuits featuring the exactly minimal number of degrees of freedom.  It is obtained by fine-tuning the Bruhat decomposition used earlier in \cite{maslov2018shorter} and leveraging the additional structure of Hadamard-free circuits. More precisely, we show that any Clifford operator can be uniquely written in the form $F_1 HS F_2$, where $H$ is a layer of Hadamard gates, $S$ is a swapping of $n$ qubits, and $F_i$ are Hadamard-free circuits chosen from suitable subgroups of the Clifford group, that depend on $H$ and $S$.  Here, $F_1$ corresponds to the -CX-CZ-P- part and $SF_2$ corresponds to the -P-CZ-CX- part in \cite{maslov2018shorter}; we optimize this layered decomposition by moving as many gates from $SF_2$ into $F_1$ as possible.  We provide a simple explicit characterization of $F_1$ and $F_2$ and give a polynomial-time algorithm for computing all layers in the above decomposition.  We describe two applications of the new canonical form.

First, we consider the problem of generating a random uniformly distributed $n$-qubit Clifford operator.  The state-of-the-art algorithm proposed by Koenig and Smolin~\cite{koenig2014efficiently} employs a sequence of $O(n)$ symplectic transvections to generate a random uniform Clifford operator.  This algorithm has runtime $O(n^3)$.  In contrast, we describe an algorithm with the runtime $O(n^2)$ that outputs a random uniformly distributed Clifford operator specified by its canonical form.  If needed, the canonical form can be converted to the stabilizer tableaux in time $O(n^\omega)$, where $\omega\,{\approx}\,2.3727$ is the matrix multiplication exponent~\cite{williams2012multiplying}.  Both Koenig-Smolin~\cite{koenig2014efficiently} and our algorithms are optimal in the sense that they consume the number of random bits that matches the  information-theoretic lower bound. We provide a Python implementation of the new algorithm (see Appendix~\ref{app:C} and Appendix~\ref{app:D}).  Our algorithm highlights a surprising connection between random uniform  Clifford operators and the Mallows distribution on the symmetric group~\cite{mallows1957non} that plays an important role in several ranking algorithms~\cite{lu2011learning}.  We show the Mallows distribution is also relevant in the context of sampling the uniform distribution on the group of invertible binary matrices.  A nearly linear-time algorithm for sampling the Mallows distribution (and its quantum generalization) is described. 

Second, we propose two new circuit decompositions.  One decomposition is useful in quantum protocols where a Clifford circuit is followed by the measurement of $n$ qubits in the computational basis. Examples of such protocols include the shadow tomography of multi-qubit states~\cite{aaronson2018shadow,huang2020predicting} or the last stage of randomized benchmarking~\cite{magesan2011scalable}.  The key observation is that the Hadamard-free operators that appear at the left/right stages of the canonical form map basis vectors to basis vectors. Applying a Hadamard-free operator immediately before the measurement of all qubits is equivalent to a simple classical postprocessing of the measurement outcomes.  Thus one of the two Hadamard-free stages in the canonical form can be skipped.  More generally, a  Clifford circuit $C$ followed by the measurement of all qubits can be replaced by another Clifford circuit $D$ as long as $DC^{-1}$ is Hadamard-free. We show that for any circuit $C$ there exists a circuit $D$ as above that contains at most $nk-\frac{k(k+1)}{2}$ two-qubit gates, where $k$ is the number of Hadamards in the canonical form of $C$.  We also show how to rewrite the canonical form in a way that reduces the number of two-qubit gate layers of the form -CZ- and -CX-. This reduced decomposition implies the ability to execute arbitrary Clifford circuits in the Linear Nearest Neighbor architecture in the two-qubit gate depth of $9n$. 

Due to the prominent role played by Hadamard-free Clifford operations as parts of Clifford circuits, we studied their circuit structure further.  Surprisingly, we discovered that in certain cases Hadamard gates can reduce the cost of implementing Hadamard-free operators.  For example, we demonstrate that so long as one is concerned with the entangling gate count (considering only $\cnotgate$ and $\czgate$ gates), certain linear reversible functions may be implemented more efficiently as Clifford circuits, as opposed to circuits relying only on the $\cnotgate$ gates.  This result can be viewed as a toy example of computational quantum advantage.  We also develop upper and lower bounds on the resource counts enabled by the introduction of Hadamard gates and prove two lemmas giving rise to two algorithms for optimizing the number of two-qubit gates in the $\cnotgate$ circuits.

We assume reader's familiarity with the concepts relating to Clifford group in the context of quantum circuits (quantum gates, Clifford tableaux, symplectic group).  Basic background information can be found in \cite{aaronson2004improved, nielsen2010quantum}, and more advanced relevant background in \cite{maslov2018shorter}. 

The rest of the paper is organized as follows. \sec{main} formally defines the canonical form of Clifford circuits, proves its existence and uniqueness, and gives an efficient algorithm for computing it.  The problem of generating random uniform Clifford operators is addressed in \sec{RandomClifford}.  Applications of the canonical form for optimization of Clifford circuits are discussed in \sec{er}. Finally, \sec{CNOT} investigates conditions under which classical reversible linear circuits can be implemented more efficiently with Clifford gates.  Appendix~\ref{app:A} addresses the problem of sampling the uniform distribution on the group $\mathrm{GL}(n)$ and highlights the role of the Mallows distribution.  Appendix~\ref{app:B} describes a version of Bruhat decomposition of the Clifford group related to the one reported in \sec{main} without explicitly developing the structure of Hadamard-free layers.  Python implementation of our algorithms can be found in Appendix~\ref{app:C} and Appendix~\ref{app:D}.

\section{Canonical form of Clifford circuits}\label{sec:main}

In this section we describe an exact parameterization of the $n$-qubit Clifford group by quantum circuits expressed using the gate set $\{ \xgate, \zgate, \pgate, \hgate, \cxgate, \czgate, \swapgate\}.$  Here $\xgate=\big(\begin{smallmatrix}0 & 1\\1 & 0 \end{smallmatrix}\big)$ and $\zgate=\big(\begin{smallmatrix}1 & 0\\0 & -1 \end{smallmatrix}\big)$ are single-qubit Pauli operators, $\pgate=\big(\begin{smallmatrix}1 & 0\\0 & i \end{smallmatrix}\big)$ is the single-qubit Phase gate, $\hgate=\frac{1}{\sqrt{2}}\big(\begin{smallmatrix}1 & 1\\1 & -1 \end{smallmatrix}\big)$ is the single-qubit Hadamard gate, and $\cxgate$ and $\czgate$ are two-qubit controlled-$\xgate$ and controlled-$\zgate$ gates correspondingly.  We denote the identity transformation $\big(\begin{smallmatrix}1 & 0\\0 & 1 \end{smallmatrix}\big)$ as $Id$.  Qubits are labeled by the integers $1,2,\ldots,n$.  We write $\cxgate^\downarrow$  to indicate that the the control qubit $c$ and the target qubit $t$ obey the relation $c\,{<}\,t$.  Our notations for the Clifford group and some of its important subgroups are summarized in the table below.

\begin{center}
\begin{tabular}{c|c|c}
\hline
{\bf Notation} & {\bf Name} & {\bf Generating set} \\
\hline
\hline
$\calC_n$ &  Clifford group & $\xgate$, $\cxgate$, $\hgate$, $\pgate$ \\
\hline
$\calF_n$ & $\hgate$-free group & $\xgate$, $\cxgate$, $\czgate$, $\pgate$ \\
\hline
$\calB_n$ & Borel group & $\xgate$, $\cxgate^\downarrow$, $\czgate$, $\pgate$ \\
\hline
$\calS_n$ & Symmetric group & $\swapgate$ \\
\hline
$\calP_n$ & Pauli group & $\xgate$, $\zgate$\\
\end{tabular}
\end{center}

First, let us explicitly describe the $\hgate$-free group $\calF_n$ and the Borel group $\calB_n\,{\subseteq}\, \calF_n$. By definition, any $\hgate$-free Clifford operator maps basis vectors to basis vectors, while possibly gaining a phase. The action of an operator $F\in \calF_n$ on a basis vector $x\in \{0,1\}^n$ can thus be compactly described as
\be
\label{eq:Hfree1}
F|x\ra = i^{x^T \Gamma x} O|\Delta x \;  (\mbox{mod $2$})\ra,
\ee
where $O\in \calP_n$ is a Pauli operator, $\Gamma,\Delta\in \mathbb{F}_2^{n\times n}$ are matrices over the binary field, $\Gamma$ is symmetric, and $\Delta$ is invertible.  Here and below we consider bit strings $x$ as column vectors, write $x^T$ for the transposed row vectors, and write $\Gamma x$, $\Delta x$ for the matrix-vector multiplication.  We denote the operator $F$ defined in \eq{Hfree1} as $F(O,\Gamma,\Delta)$.  This operator is said to have trivial Pauli part if $O\,{=}\,Id$.

An operator $F(O,\Gamma,\Delta)$ belongs to the Borel group $\calB_n$, iff the matrix $\Delta$ is lower-triangular and unit-diagonal, i.e., $\Delta_{i,j}=0$ for $i\,{<}\,j$ and $\Delta_{i,i}=1$ for all $i$.  Note that any lower-triangular unit-diagonal matrix $\Delta$ is automatically invertible.  Any element of the Borel group admits representation by a canonical quantum circuit, 
\be
\label{eq:Borel1}
F(O,\Gamma,\Delta) = O \prod_{i=1}^n  \pgate_i^{\Gamma_{i,i}}
\prod_{1\le i<j\le n} \czgate_{i,j}^{\Gamma_{i,j}} \prod_{1\le i<j\le n} \cxgate_{i,j}^{\Delta_{j,i}}.
\ee 
Here the product of $\cxgate$ gates is ordered such that the control qubit
index increases from the left to the right.  For example, for $n{=}4$ the last product in \eq{Borel1} is
\[
\cxgate_{1,2}^{\Delta_{2,1}}  \cxgate_{1,3}^{\Delta_{3,1}}  \cxgate_{1,4}^{\Delta_{4,1}}
\cxgate_{2,3}^{\Delta_{3,2}}  \cxgate_{2,4}^{\Delta_{4,2}} \cxgate_{3,4}^{\Delta_{4,3}}.
\]  
By counting the number of bits needed to specify the data set $\{O,\Gamma,\Delta\}$ one gets $|\calB_n|=2^{n^2+2n}$ (we ignore the overall phase of Clifford operators). Given a permutation of qubits $S\in \calS_n$, we write $j{=}S(i)$ if $S$ maps the $i$-th qubit to the $j$-th qubit.  Given integers $a\le b$, we write $\irange{a}{b}$ to denote the set of integers
$i$ such that $a\le i\le b$.

\begin{theorem}[\bf Canonical Form]
\label{thm:CliffordExact}
Any  Clifford operator $U\in \calC_n$ can be uniquely written as
\be
\label{eq:CliffordExactParam}
U = F(Id,\Gamma,\Delta) \cdot \left( \prod_{i=1}^n \hgate_i^{h_i} \right) S \cdot F(O',\Gamma',\Delta')
\ee
where $h_i\,{\in}\, \{0,1\}$, $S\,{\in}\, \calS_n$ is a permutation of $n$ qubits, and $F(Id,\Gamma,\Delta)$, $F(O',\Gamma',\Delta')$ are elements of the Borel group $\calB_n$ such that the matrices $\Gamma$, $\Delta$ obey the following rules for all $i,j\,{\in}\, \irange{1}{n}$:
\begin{enumerate}
\item[\bf C1] if $h_i=0$ and $h_j=0$,   then $\Gamma_{i,j}=0$;
\item[\bf C2] if $h_i=1$  and $h_j=0$  and $S(i)>S(j)$, then $\Gamma_{i,j}=0$;
\item[\bf C3] if  $h_i=0$  and $h_j=0$ and  $S(i)>S(j)$, then $\Delta_{i,j}=0$;
\item[\bf C4] if  $h_i=1$ and  $h_j=1$ and $S(i)<S(j)$, then $\Delta_{i,j}=0$;
\item[\bf C5] if  $h_i=1$ and $h_j=0$, then $\Delta_{i,j}=0$.
\end{enumerate}
The canonical form \eq{CliffordExactParam} can be computed in time $poly(n)$, given the stabilizer tableaux of $U$. 
\end{theorem}
As a simple example, consider the following circuit identity
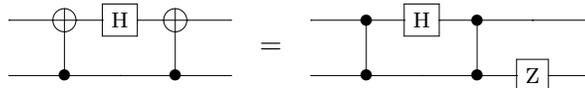
\[
{\Qcircuit @C=0.7em @R=0.7em @!{
&\targ	& \gate{\hgate} 	& \targ 	& \qw \\ 
&\ctrl{-1} & \qw & \ctrl{-1} & \qw }}
\hspace{1em}\raisebox{-0.9em}{=}\hspace{1em}
{
\Qcircuit @C=0.7em @R=0.7em @!{
&  \control \qw	& \gate{\hgate} 	& \control \qw	& \qw & \qw \\ 
&\ctrl{-1} & \qw & \ctrl{-1} & \gate{\zgate} & \qw
}}
\]
It can be verified by noting that $\xgate \hgate \xgate \,{=}\, - \zgate \hgate \zgate$.  The circuit on the left is not canonical since the $\cxgate$ gates have a wrong order of the control and target qubits. The circuit on the right is canonical with the identity permutation layer, $S(1)\,{=}\,1$, $S(2)\,{=}\,2$, the Hadamard layer $h\,{=}\,10$, 
\[
\Gamma = \Gamma'=\left[\ba{cc} 0 & 1\\1 & 0\\ \ea\right],
\quad
\Delta =\Delta' = \left[\ba{cc} 1 & 0\\0 & 1\\ \ea\right],
\]
and $O'\,{=}\,\zgate_2$.  Direct inspection shows that all the rules C1-C5 are satisfied. More generally, the rules C1-C5 describe the ways in which the individual $\pgate$ (C1 for $i{=}j$), $\czgate$ (C1-C2), and $\cnotgate$ (C3-C5) gates in the canonical description of $F(Id,\Gamma,\Delta)$ per \eq{Borel1} can be moved through the Hadamard-SWAP stage $hS$ in \eq{CliffordExactParam} to the right, and in doing so are transformed into a gate that belongs to the Borel group, i.e. $\pgate$, or $\czgate$, or $\cxgate^\downarrow$. Such gates can be absorbed into the stage $F(O',\Gamma',\Delta')$. The decomposition from \eq{CliffordExactParam} is utilized for the generation of random uniformly distributed Clifford operators. Indeed, in \sec{RandomClifford} we show that $U$ is uniformly distributed if the pair $h,S$ is sampled from a suitable generalization of the  Mallows distribution on the symmetric group~\cite{mallows1957non}.  For fixed $h$ and $S$, the matrices $\Gamma,\,\Delta,\,\Gamma'$, and $\Delta'$ have to be sampled uniformly subject to the rules C1-C5.  

We will divide the proof of \thm{CliffordExact} into two parts. The first part, presented in Subsection~\ref{ssec:part1}, proves the existence and the uniqueness of the canonical form \eq{CliffordExactParam}. An efficient algorithm for computing the canonical form is described in Subsection~\ref{ssub:part2}.
 
\subsection{The existence and the uniqueness of the canonical form}
\label{ssec:part1}

For any operator $W\,{\in}\,\calC_n$ define a set $\calB_n W \calB_n := \{ FWF'{:} \; F,F' \in \calB_n\}$. Our starting point is the Bruhat decomposition  of the Clifford group \cite[Section~IV]{maslov2018shorter}.

\vspace{1mm}\noindent {\bf Bruhat decomposition.}{\it
The Clifford group $\calC_n$ is a disjoint union 
\be
\label{eq:bruhat001}
\calC_n =  \bigsqcup_{h \in \{0,1\}^n} \;  \bigsqcup_{S\in \calS_n}
 \calB_n  \left( \prod_{i=1}^n \hgate_i^{h_i} \right) S \; \calB_n.
\ee
}
It follows that any Clifford operator can be written as $U{=}FWF'$, where $F,\,F'\in \calB_n$ and $W$ is a layer of Hadamards $\hgate$ on qubits $i$ such that $h_i{=}1$ followed by a qubit permutation $S$.  The disjointness of the union in \eq{bruhat001} implies that $W$ is uniquely defined by $U$. However, the operators $F$ and $F'$ are generally non-unique.  For example, if $W\,{=}\,Id$ then $U\,{=}\,FF'$. In this case, one can arbitrarily assign the gates to either $F$ or $F'$, so long as their product evaluates to the desired $U$. 

We will prove that the Bruhat decomposition $U{=}FWF'$ becomes unique if we restrict $F$ to a certain subgroup of $\calB_n$ that depends on $h$ and $S$. From now on we focus on some fixed pair $h\,{\in}\, \{0,1\}^n$ and $S\,{\in}\, \calS_n$. Define the operator
\be
\label{eq:W}
W= \left( \prod_{i=1}^n \hgate_i^{h_i} \right) S
\ee
and a group
\be
\label{eq:bruhat003}
\calB_n(h,S) = \{ F \in \calB_n{:} \; W^{-1} F W \in \calB_n\}.
\ee
Note that $\calB_n(h,S)$ is a group since it is the intersection of two groups $\calB_n$ and $W\calB_n W^{-1}$. We will need two technical lemmas, both proved at the end of this subsection.  The first lemma clarifies the rules C1-C5.
\begin{lemma}
\label{lem:FHS}
Let $\bar{h}=h\oplus 1^n$ be the bitwise negation of $h$. Then
\be
\label{eq:HHbar}
\calB_n(\bar{h},S)  \cap \calB_n(h,S)  = \calP_n.
\ee
An operator $F(O,\Gamma,\Delta)\in \calB_n$  belongs to the group $\calB_n(\bar{h},S)$ if and only if $\Gamma$ and $\Delta$  obey the rules C1-C5.  
\end{lemma}
Recall that $\calP_n$ denotes the Pauli group. The second lemma provides a unique decomposition of any operator $F\,{\in}\, \calB_n$ in terms of the elements of the groups $\calB_n(\bar{h},S)$ and $\calB_n(h,S)$.
\begin{lemma}
\label{lem:HHbar}
Any operator $F\,{\in}\, \calB_n$ can be uniquely written as $F=F_L F_R$ for some $F_L \in \calB_n(\bar{h},S)$ and some $F_R \,{\in}\,  \calB_n(h,S)$ such that $F_L$ has trivial Pauli part.
\end{lemma}

Let us now prove the theorem. Using Bruhat decomposition, \eq{bruhat001}, write $U\,{=}\,LWR$ for some $L,\,R\in \calB_n$ and some $W$ defined in \eq{W}.  The disjointness of the union in \eq{bruhat001} implies that $W$ is uniquely defined.  Using \lem{HHbar} write $L = B C$ for some $B\,{\in}\, \calB_n(\bar{h},S)$ and $C\,{\in}\,  \calB_n(h,S)$ such that $B$ has trivial Pauli part.  Then
\[
U=LWR=BCWR=B W W^{-1} C W R = B W C'R,
\]
where $C'=W^{-1} C W\in \calB_n$ by definition of the subgroup $\calB_n(h,S)$, see \eq{bruhat003}. Setting $B'=C'R$ we get the desired decomposition \eq{CliffordExactParam}, namely, $U=BWB'$.  The inclusion $B\,{\in}\, \calB_n(\bar{h},S)$ and the assumption that $B$ has trivial Pauli part are equivalent to conditions C1-C5 due to \lem{FHS}.

It remains to check that this decomposition is unique.  Suppose 
\be
\label{eq:eq_unique}
BWB'=CWC'
\ee
for some $C\,{\in}\,\calB_n(\bar{h},S)$ and $C'\,{\in}\, \calB_n$ such that $C$ has trivial Pauli part.  We need to establish that in this case, $B\,{=}\,C$ and $B'\,{=}\,C'$. Indeed, \eq{eq_unique} gives
\be
\label{eq:eq_unique1}
W^{-1} (C^{-1} B) W= C'(B')^{-1} \in \calB_n.
\ee
Here the inclusion $C'(B')^{-1} \in \calB_n$ follows from the assumption $B',\,C'\in \calB_n$ and  the fact that $\calB_n$ is a group. From \eq{bruhat003} and \eq{eq_unique1} it follows that $C^{-1} B \in \calB_n(h,S)$.  However,  $B$ and $C$ are assumed to be in $\calB_n(\bar{h},S)$.  Since the latter is a group, $C^{-1} B$ is also in $\calB_n(\bar{h},S)$. We conclude that 
\[
C^{-1} B \in \calB_n(h,S) \cap  \calB_n(\bar{h},S) = \calP_n,
\]
that is, $C^{-1} B$ is a Pauli operator. However, $B$ and $C$ are assumed to have trivial Pauli part. Thus, $C^{-1}B=Id$, implying $B{=}C$.  From \eq{eq_unique} one obtains $B'{=}C'$, as claimed.

In the rest of this subsection we prove \lem{FHS} and \lem{HHbar}.

\begin{proof}[\bf Proof of Lemma~\ref{lem:FHS}]
It is more convenient to consider the group $\calB_n(h,S)$ instead of $\calB_n(\bar{h},S)$.  As above, define
\[
W= \left( \prod_{i=1}^n \hgate_i^{h_i} \right) S.
\]
Let $F\,{=}\,F(O,\Gamma,\Delta)$ be some element of $\calB_n$, see \eq{Hfree1} and \eq{Borel1}.  Recall that $\Gamma$ is a symmetric binary matrix, while $\Delta$ is a lower-triangular unit-diagonal matrix. We need to prove that $W^{-1} F W \in \calB_n$ if and only if  the following rules hold for all $i,j\,{\in}\, \irange{1}{n}$:
\begin{enumerate}
\item[\bf B1] if $h_i=1$ and $h_j=1$,   then $\Gamma_{i,j}=0$;
\item[\bf B2] if $h_i=0$  and $h_j=1$  and $S(i)>S(j)$, then $\Gamma_{i,j}=0$;
\item[\bf B3] if  $h_i=1$  and $h_j=1$ and  $S(i)>S(j)$, then $\Delta_{i,j}=0$;
\item[\bf B4] if  $h_i=0$ and  $h_j=0$ and $S(i)<S(j)$, then $\Delta_{i,j}=0$;
\item[\bf B5] if  $h_i=0$ and $h_j=1$, then $\Delta_{i,j}=0$.
\end{enumerate}
These rules are obtained from C1-C5 respectively by negating each bit of $h$. Suppose first that $F$ is a single gate from the generating set  $\{\xgate,\zgate,\pgate,\czgate,\cxgate^\downarrow\}$ of $\calB_n$, see \eq{Borel1}.

\noindent
{\em Case~1:}  $F{=}\xgate_i$ or $F{=}\zgate_i$.  Note that B1-B5 impose no restrictions on the Pauli part of $F$.  Since $W$ is a Clifford operator, $W^{-1} F W$ is a Pauli operator. As the Pauli group is a subgroup of $\calB_n$, we infer that $W^{-1} F W \in \calB_n$.

\noindent
{\em Case~2:}  $F{=}\pgate_i$. Then $W^{-1} F W{=}\pgate_{S(i)}$ if $h_i{=}0$ and $W^{-1} F W=\hgate_{S(i)} \pgate_{S(i)} \hgate_{S(i)}$ if $h_i\,{=}\,1$. Direct inspection of the unitary matrix reveals that the operator $\hgate \pgate \hgate$ is not $\hgate$-free. Thus, $W^{-1} F W \in \calB_n$ iff $h_i\,{=}\,0$. This gives the rule B1 with $i{=}j$. Indeed, $\pgate_i=F(I,\Gamma,0^{n\times n})$ where $\Gamma$ has a single non-zero element $\Gamma_{i,i}{=}1$.

\noindent
{\em Case~3:}  $F{=}\czgate_{i,j}$. 
The operator $W^{-1} F W$ depends on the bits $h_i$ and $h_j$ as shown in the
following table.
\begin{center}
\begin{tabular}{c|c|c}
\hline
$h_i$ & $h_j$ & $W^{-1} \czgate_{i,j} W$\\
\hline
$0$ & $0$ & $\czgate_{S(i),S(j)}$ \\
$0$ & $1$ & $\cxgate_{S(i),S(j)}$  \\
$1$ & $0$ & $\cxgate_{S(j),S(i)}$ \\
$1$ & $1$ & $(\hgate\otimes \hgate \cdot \czgate\cdot \hgate \otimes \hgate)_{S(i),S(j)}$ \\
\end{tabular}
\end{center}
If $h_i{=}h_j{=}0$, then $W^{-1} F W\in \calB_n$ regardless of $S$.  If $h_i{=}0$ and $h_j{=}1$, then $W^{-1} F W\in \calB_n$ iff $S(i){<}S(j)$ since $\calB_n$ only includes $\cxgate^\downarrow$ gates.  This gives the rule B2. Likewise, if $h_i{=}1$ and $h_j{=}0$, then $W^{-1} F W\in \calB_n$ iff $S(j){<}S(i)$. This gives the rule B2 with $i$ and $j$ exchanged.  Direct inspection shows that the operator $\hgate\otimes \hgate \cdot \czgate\cdot \hgate \otimes \hgate$ is not $\hgate$-free. Thus, if $h_i{=}h_j{=}1$ then $W^{-1} F W$ is not in $\calB_n$ regardless of $S$. This gives the rule B1 with $i{\ne}j$.

\noindent
{\em Case~4:}  $F{=}\cxgate_{i,j}$. Note that $i{<}j$ since we assumed $F\,{\in}\, \calB_n$.
Recall that $F(O,\Gamma,\Delta)$ includes $\cxgate_{i,j}$ iff $\Delta_{j,i}\,{=}\,1$, see \eq{Borel1}.
The operator $W^{-1} F W$ depends on the bits $h_i$ and $h_j$ as shown in the
following table.
\begin{center}
\begin{tabular}{c|c|c}
\hline
$h_i$ & $h_j$ & $W^{-1} \cxgate_{i,j} W$\\
\hline
$0$ & $0$ & $\cxgate_{S(i),S(j)}$ \\
$0$ & $1$ & $\czgate_{S(i),S(j)}$  \\
$1$ & $0$ & $(\hgate\otimes \idgate \cdot \cxgate\cdot \hgate \otimes \idgate)_{S(j),S(i)}$ \\
$1$ & $1$ & $\cxgate_{S(j),S(i)}$ \\
\end{tabular}
\end{center}
If $h_i{=}h_j{=}0$, then $W^{-1} F W\in \calB_n$ iff $S(i){<}S(j)$. This gives the rule B4.  Direct inspection shows that the operator $\hgate\otimes \idgate \cdot \cxgate\cdot \hgate \otimes \idgate$ is not $\hgate$-free. Thus,  if $h_i{=}1$ and $h_j{=}0$, then $W^{-1} F W$ is not in $\calB_n$ regardless of $S$.  This gives the rule B5.  Finally, if $h_i{=}h_j{=}1$, then $W^{-1} F W\in \calB_n$ iff $S(j){<}S(i)$. This gives the rule B3.

Consider a general operator $F{=}F(O,\Gamma,\Delta)\in \calB_n$ such that $\Gamma$ and $\Delta$ obey the rules B1-B5. We claim that $W^{-1} F W\in \calB_n$. Indeed, using \eq{Borel1} one can write $F\,{=}\,F_L \cdots F_2 F_1$, where $F_i$ are individual gates from the gate set $\{\xgate,\zgate,\pgate,\czgate,\cxgate^\downarrow\}$. Since we have already checked the rules B1-B5 for each individual gate, one gets $W^{-1} F_i W \in \calB_n$. Since $W^{-1} F W$ is a product of operators $W^{-1} F_i W$, one infers that $W^{-1} F W\in \calB_n$, as claimed. In other words, the rules B1-B5 are sufficient for the inclusion $W^{-1} F W\in \calB_n$.

It remains to check that the rules B1-B5 are also necessary.
Let $\calB_n'(h,S)$ be the set of all operators $F\,{\in}\, \calB_n$ that obey
B1-B5. The above shows that 
\be
\label{eq:B'B}
\calB_n'(h,S) \subseteq \calB_n(h,S).
\ee

We next prove that $\calB_n'(h,S)\,{=}\,\calB_n(h,S)$ for all pairs $h$ and $S$ by employing the counting argument.  Indeed, Bruhat decomposition implies that any Clifford operator $U$ can be written (possibly non-uniquely) as $U\,{=}\,FWF'$, where $F,\,F'\in \calB_n$ and $F$ is some canonical representative of the left coset $F\calB_n(h,S)$.  If $F\calB_n(h,S) = G\calB_n(h,S)$ for some $G\,{\in}\, \calB_n$, then $F\,{=}\,GL$ for some $L\,{\in}\, \calB_n(h,S)$. Thus $U\,{=}\,GWG'$, where $G'=(W^{-1}L W)F'$. Note that $G'\,{\in}\, \calB_n$ since both $W^{-1}L W$ and $F'$ are in $\calB_n$. Thus the number of triples $(W,F,F')$ must be at least $|\calC_n|$. Using this observation and \eq{B'B} one gets
\be
\label{eq:Csize}
|\calC_n|\le \sum_{h\in \{0,1\}^n} \; \sum_{S\in \calS_n} \; \frac{|\calB_n|^2}{|\calB_n(h,S)|}
 \le \sum_{h\in \{0,1\}^n} \; \sum_{S\in \calS_n} \; \frac{|\calB_n|^2}{|\calB_n'(h,S)|}.
\ee
If we can show that the right-hand side of \eq{Csize} coincides with $|\calC_n|$, this would imply $\calB_n'(h,S) =\calB_n(h,S)$ for all $h$ and $S$. By the definition of $\calB_n'(h,S)$, one has
\be
\label{eq:count1}
|\calB_n'(h,S)|= |\calB_n| \cdot 2^{-I_n(h,S)}, 
\ee
where $I_n(h,S)$ is the number of  independent constraints on $\Gamma$ and $\Delta$ imposed by the rules B1-B5 for a given $h$ and $S$.  The number of constraints enforced by each individual rule is summarized in the following table.
\begin{center}
\begin{tabular}{r|l}
\hline
 & Number of constraints on $\Gamma,\Delta$ \\
\hline
\hline
{\bf B1} &  $|h| + \sum_{i>j} h_i h_j$ \\
{\bf B2} & $\sum_{i>j,\; S(i)>S(j)} \bar{h}_i h_j + \sum_{i>j,\; S(i)<S(j)} h_i \bar{h}_j$ \\
{\bf B3} & $\sum_{i>j,\; S(i)>S(j)} h_i h_j$ \\
{\bf B4} & $\sum_{i>j,\; S(i)<S(j)} \bar{h}_i \bar{h}_j$ \\
{\bf B5} & $\sum_{i>j}  \bar{h}_j h_i$ \\
\end{tabular}
\end{center}
Summing up all terms in the table gives
\be
\label{eq:count2}
I_n(h,S) = n(n{-}1)/2 +|h| +  \sum_{\substack{1\le i<j\le n\\ S(i)<S(j) \\ }}\;  (-1)^{1+ h_i}.
\ee
Recall that $|\calC_n|=2^{n^2+2n}\prod_{i=1}^n (4^i{-}1)$ and $|\calB_n|=2^{n^2+2n}$.
Thus it suffices to check that 
\be
\label{eq:count3}
F(n):= \sum_{h\in \{0,1\}^n} \; \sum_{S\in \calS_n} 2^{I_n(h,S)} = \prod_{i=1}^n (4^i-1).
\ee
We use the induction in $n$. The base of induction is $n{=}1$. In this case, $I_1(h,S)=I_1(h)$ since $\calS_1$ has only one element, the identity. From \eq{count2} one gets $I_1(0){=}0$ and $I_1(1){=}1$. Thus $F(1)=2^{I_1(0)} + 2^{I_1(1)} =3$, as claimed in \eq{count3}.  

Next, show the induction step. Assume \eq{count3} holds for $F(n{-}1)$. Write $h=(h_1,h')$ with $h'\,{\in}\, \{0,1\}^{n-1}$. Let $m=S(1)$ and $S'\,{\in}\, \calS_{n-1}$ be a permutation of integers $\irange{2}{n}$ such that $S'(j) = S(j)+1$ if $S(j)<m$ and $S'(j)=S(j)$ if $S(j)>m$.  Simple algebra gives
\[
I_n(h,S) = I_{n-1}(h',S') + n{-}1 + h_1 + (n{-}m)(-1)^{1+h_1}.
\] 
Furthermore, the map $S\to (m,S')$ is a one-to-one map $\calS_n \to \irange{1}{n} \, {\times}\, \calS_{n-1}$. Thus
\begin{eqnarray*}
F(n) = F(n{-}1) \sum_{m=1}^n \; \sum_{h_1=0,1} 2^{ n-1 + h_1 + (n-m)(-1)^{1+h_1}} \\
=(4^n-1) F(n{-}1).
\end{eqnarray*}
This proves \eq{count3}. Combining Eqs.~(\ref{eq:count1},\ref{eq:count3}) one gets 
\begin{eqnarray*}
\sum_{h\in \{0,1\}^n} \; \sum_{S\in \calS_n} \; \frac{|\calB_n|^2}{|\calB_n'(h,S)|} = |\calB_n| F(n) \\
= 2^{n^2+2n}
\prod_{i=1}^n (4^i-1) = |\calC_n|.
\end{eqnarray*}
Substituting this into \eq{Csize} one concludes that $\calB_n'(h,S)\,{=}\,\calB_n(h,S)$ for all $h$ and $S$, implying that $\calB_n'(h,S)$ is a group since $\calB_n(h,S)$ is.  In other words, the rules B1-B5 are necessary and sufficient for the inclusion $W^{-1} FW\in \calB_n$.

Finally, let us finish the proof by establishing \eq{HHbar}.  Direct inspection shows that the only matrices $\Gamma$ and $\Delta$ that obey the rules B1-B5 for both $h$ and $\bar{h}$ are all-zero matrices.  Thus  $\calB_n'(h,S) \cap \calB_n'(\bar{h},S)=\calP_n$. However, we have already proved that $\calB_n'(h,S)=\calB_n(h,S)$ and $\calB_n'(\bar{h},S)=\calB_n(\bar{h},S)$.  This proves \eq{HHbar}.
\end{proof}

\begin{proof}[\bf Proof of Lemma~\ref{lem:HHbar}]
Let $F_1,F_2,\ldots,F_m$ be the list of all elements of $\calB_n(\bar{h},S)$ with trivial Pauli part.  Note that 
\[
m=\frac{|\calB_n(\bar{h},S)|}{|\calP_n|} = 4^{-n} |\calB_n(\bar{h},S)|.
\]
We claim that the right cosets $F_j\calB_n(h,S)$ are pairwise disjoint. Indeed, suppose $F_i \calB_n(h,S) \cap F_j \calB_n(h,S)\ne \emptyset$ for some $F_i \ne F_j$. Then $F_j^{-1} F_i \in \calB_n(h,S)$. From \eq{HHbar} one infers that $F_j F_i^{-1}$ is a Pauli operator.  However, by assumption, $F_i$ and $F_j$ have trivial Pauli part. Thus $F_i{=}F_j$, which is a contradiction.

To conclude the proof it suffices to show that 
\be
\label{eq:total_size}
|\calB_n(h,S)|\cdot |\calB_n(\bar{h},S)| = |\calP_n|\cdot  |\calB_n|= 4^n |\calB_n|.
\ee
Indeed, if this is the case, then the number of 
elements of $\calB_n$ that belong to some
right coset $F_j \calB_n(h,S)$ is
\[
m\cdot |\calB_n(h,S) | =  4^{-n} |\calB_n(\bar{h},S)| \cdot  |\calB_n(h,S)| =  |\calB_n|,
\]
that is, any element of $\calB_n$  is contained in some coset $F_j \calB_n(h,S)$. From \eq{count1} one gets $|\calB_n(h,S)|=|\calB_n| 2^{-I_n(h,S)}$ and $|\calB_n(\bar{h},S)|=|\calB_n| 2^{-I_n(\bar{h},S)}$, where the function $I_n(h,S)$ is defined in \eq{count2}.  Simple algebra gives $I_n(h,S)+I_n(\bar{h},S)=n^2$.  Recalling that $|\calB_n|=2^{n^2+2n}$ gives
\begin{eqnarray*}
 |\calB_n(\bar{h},S)| \cdot  |\calB_n(h,S)| = |\calB_n|^2 \cdot 2^{-I_n(h,S)-I_n(\bar{h},S)} \\
= |\calB_n|^2  2^{-n^2} = 2^{n^2 +4n} = 4^n |\calB_n|,
\end{eqnarray*}
proving \eq{total_size}.
\end{proof}

\subsection{How to compute the canonical form}
\label{ssub:part2}

First, let us introduce some terminology. Suppose $U\,{\in}\, \calC_n$ is a Clifford operator. We assume that $U$ is specified by its stabilizer tableaux~\cite{aaronson2004improved}, defined as a list of $2n$ Pauli operators $U\xgate_iU^{-1}$ and $U\zgate_i U^{-1}$.  It is well-known that the stabilizer tableaux uniquely specifies $U$ up to an overall phase~\cite{aaronson2004improved}.  We say that $U$ is  {\em non-entangling} if $U\xgate_i U^{-1}$ and $U\zgate_i U^{-1}$ are single-qubit Pauli operators for all $i\,{\in}\, \irange{1}{n}$.  Given integers $i,j\,{\in}\, \irange{1}{n}$, let us say that $U$ is {\em $(i,j)$-non-entangling} if $U\xgate_i U^{-1}$ and $U\zgate_i U^{-1}$ are single-qubit Pauli operators acting on the $j$-th qubit. 
\begin{lemma}[\bf Non-entangling Clifford operators]
\label{lem:algo1}
Any non-entangling operator $U\,{\in}\,\calC_n$ has the form
\be
\label{eq:nonentangling}
U=F_1 \left(\prod_{i=1}^n \hgate_i^{h_i} \right) S F_2
\ee
where $h\,{\in}\, \{0,1\}^n$, $S\,{\in}\, \calS_n$, and $F_1,F_2\in \calB_n$ are tensor products of single-qubit $\hgate$-free Clifford operators.  This decomposition can be computed in time $O(n^2)$.
\end{lemma}
\begin{proof}
Suppose $\tilde{X}_i=U\xgate_iU^{-1}$ and $\tilde{Z}_i=U\zgate_iU^{-1}$ are single-qubit Pauli operators.  Since $\tilde{X}_i$ anti-commutes with $\tilde{Z}_i$, they must act on the same qubit.  Let this qubit be $S(i)$, where $S\,{:}\, \irange{1}{n} \to \irange{1}{n}$ is some function.  We claim that $S$ is a permutation. Indeed, otherwise $S(i){=}S(j){=}k$ for some $i{\ne} j$.  Then $\tilde{X}_i$, $\tilde{Z}_i$, $\tilde{X}_j$, $\tilde{Z}_j$ are independent  Pauli operators acting on the $k$-th qubit.  This is impossible since there are only two independent single-qubit Pauli operators. Thus $S$ is a permutation of $n$ qubits. Let $V\,{=}\,S^{-1} U$.  Then $V\xgate_i V^{-1}, V\zgate_iV^{-1}\in \{\xgate_i,\ygate_i,\zgate_i\}$ for all $i$. Thus $V$ is a product of single-qubit Clifford operators, $V=\prod_{i=1}^n V_i$.  Any single-qubit Clifford operator $V_i$ can be written as
\[
V_i = \pgate^{a_i}_i \hgate^{b_i}_i \pgate^{c_i}_i \xgate^{d_i}_i \zgate^{e_i}_i
\quad \mbox{for some} \quad a_i,b_i,c_i,d_i,e_i\in \{0,1\}.
\]
Writing $U=SV$
and commuting all single-qubit gates $\pgate^{a_i}_i$ 
and $\hgate^{b_i}_i$ to the left one gets
\[
U = \left( \prod_{i=1}^n \pgate^{a_i}_{S(i)} \hgate^{b_i}_{S(i)} \right)
S \left( \prod_{i=1}^n \pgate^{c_i}_i \xgate^{d_i}_i \zgate^{e_i}_i \right).
\]
This is the desired decomposition Eq.~(\ref{eq:nonentangling})
with $h_{S(i)}=b_i$.
Clearly, all above steps can be performed in time $O(n^2)$, given the
stabilizer tableaux of $U$.
\end{proof}
Below we give an algorithm that takes as input a Clifford operator $U\,{\in}\, \calC_n$ and computes $B_1,B_2 \in \calB_n$ such that $B_1 U B_2$ is non-entangling. This is achieved by a sequence of elementary steps that ``disentangle" one qubit per time.  More formally, the first $m$ disentangling steps provide operators $B_1,B_2\in \calB_n$ such that $B_1 U B_2$ is $(j,k_j)$-non-entangling for $j\,{\in}\,\irange{1}{m}$ and some $m$-tuple of distinct integers $k_1,k_2,\ldots,k_m\in \irange{1}{n}$.  The full algorithm takes time $O(n^3)$.  First consider a simpler task of ``disentangling" a Pauli operator.
\begin{lemma}[\bf Disentangling a Pauli operator]
\label{lem:algo2}
Given a Pauli operator $O\,{\in}\, \calP_n$,  there exists $B\,{\in}\, \calB_n$ such that 
$BOB^{-1}$ is a single-qubit Pauli operator. 
One can compute both $B$ and $BOB^{-1}$ in time $O(n)$.
\end{lemma}
\begin{proof}
Write 
\[
O = \prod_{j=1}^n \xgate_j^{\alpha_j} \zgate_j^{\beta_j},
\]
where $\alpha,\beta\in \{0,1\}^n$. We can assume without loss of generality that either $\alpha\ne 0^n$ or $\beta\ne 0^n$ (otherwise choose $B=Id$). Suppose first that $\alpha\ne 0^n$.  Let $i\,{\in}\, \irange{1}{n}$ be the first non-zero element of $\alpha$.  Define an operator
\[
B_1 = \prod_{j=i+1}^n \cxgate_{i,j}^{\alpha_j}.
\]
Using the identities $\cxgate_{i,j} \xgate_i \cxgate_{i,j}=\xgate_i \xgate_j$ and $\cxgate_{i,j} \zgate_j \cxgate_{i,j} = \zgate_i \zgate_j$ one gets
\[
B_1OB_1^{-1}= \xgate_i \zgate_i^\epsilon \prod_{j=1}^n \zgate_j^{\beta_j},
\quad \epsilon = \sum_{j=i+1}^n \alpha_j \beta_j {\pmod 2}.
\]
Define an operator
\[
B_2 = \prod_{j\in \irange{1}{n}\setminus i} \czgate_{i,j}^{\beta_j}.
\]
Using the identity $\czgate_{i,j} \xgate_i \czgate_{i,j}=\xgate_i \zgate_j$ one gets $B_2 B_1 O B_1^{-1} B_2^{-1} = \xgate_i \zgate_i^\epsilon$.  Thus the desired operator $B$ can be chosen as $B=B_2 B_1$.

Suppose now that $\alpha\,{=}\,0^n$. Then  $\beta\,{\ne}\, 0^n$.  Let $j\,{\in}\, \irange{1}{n}$ be the largest qubit index such that $\beta_j\,{=}\,1$.  Define an operator
\[
B = \prod_{i=1}^{j-1} \cxgate_{i,j}^{\beta_i}.
\]
Then $BOB^{-1} = \zgate_j$. Clearly, all the above steps take time $O(n)$.  In both cases the  operator $B$ is expressed using gates $\cxgate^{\downarrow}$ and $\czgate$. Thus $B\,{\in}\, \calB_n$.
\end{proof}

\begin{lemma}[\bf Disentangling a Clifford operator]
\label{lem:algo3}
For any $U\,{\in}\, \calC_n$ there exist  $k\,{\in}\, \irange{1}{n}$ and  $B_1,B_2 \in \calB_n$ such that $B_1 U B_2$ is $(1,k)$-non-entangling.  One can compute $(k,B_1,B_2)$ and $B_1 U B_2$  in time $O(n^2)$.
\end{lemma}
\begin{proof}
Let $O=U\zgate_1U^{-1}$. By \lem{algo2}, there exists
$B_1\,{\in}\, \calB_n$ such that $B_1OB_1^{-1}$ is a single-qubit Pauli operator acting on some qubit $k\,{\in}\, \irange{1}{n}$.  Perform an update $U\gets B_1 U$. Then $U\zgate_1U^{-1}$ acts only on the $k$-th qubit. Thus we can write
\be
\label{eq:disent1}
U\zgate_1 U^{-1} = O_k\otimes Id\,^{n-1}
\ee
for some $O\,{\in}\, \calP_1$. Here the tensor product separates the $k$-th qubit and the remaining qubits.

Note that $U \zgate_i U^{-1}$ and $U\xgate_iU^{-1}$ with $i\,{\ge}\, 2$ may act on the $k$-th qubit only by the identity or by $O$.  Indeed, these operators commute with $U\zgate_1 U^{-1}$ which has the form Eq.~(\ref{eq:disent1}).  Define $B_2 := \prod_{i\in Q} \cxgate_{1,i}$, where $Q$ is the set of qubits $2\,{\le}\, i\,{\le}\, n$ such that $U\zgate_i U^{-1}$ acts non-trivially on the $k$-th qubit.  Then $(UB_2) \zgate_i (UB_2)^{-1}$ acts trivially on the $k$-th qubit for $i\,{\ge}\, 2$. Furthermore, $(UB_2) \zgate_1 (UB_2)^{-1} = U\zgate_1 U^{-1}$ since $B_2$ commutes with $\zgate_1$.  Perform an update $U\gets UB_2$. The updated $U$ obeys Eq.~(\ref{eq:disent1}) and
\be
\label{eq:disent2}
U\zgate_i U^{-1} = I_k \otimes \tilde{Z}_i  \mbox{ and }
U\xgate_i U^{-1} = O_k^{\epsilon_i} \otimes \tilde{X}_i,
\quad 2\le i\le n
\ee
for some $\epsilon_i\,{\in}\, \{0,1\}$, and some Pauli operators $\tilde{Z}_i, \tilde{X}_i \in \calP_{n-1}$ that obey the same commutation rules as the Pauli operators $\zgate_i$, $\xgate_i$ on $n{-}1$ qubits. We claim that 
\be
\label{eq:disent3}
U\xgate_1U^{-1} = R_k \otimes \prod_{i=2}^n \tilde{Z}_i^{\epsilon_i}
\ee
for some Pauli operator $R_k\,{\in}\, \calP_1$ such that $OR_k=-R_kO$. Indeed, one can check that the operator defined in Eq.~(\ref{eq:disent3}) commutes with all operators $U\zgate_i U^{-1}$ and $U\xgate_i U^{-1}$ for $2\le i\le n$ and anti-commutes with $U\zgate_1 U^{-1}$.  Define
\[
B_3 = \prod_{i=2}^n \czgate_{1,i}^{\epsilon_i}.
\]
Clearly, $B_3$ commutes with $\zgate_i$ for all $i$. Furthermore, $B_3\xgate_1 B_3^{-1}= \xgate_1 \prod_{i=2}^n \zgate_i^{\epsilon_i}$.  Perform an update $U\gets UB_3$. From Eqs.~(\ref{eq:disent2},\ref{eq:disent3}) one gets $U\xgate_1 U^{-1} =R_k\otimes I^{n-1}$.  Combining this and Eq.~(\ref{eq:disent1}) one concludes that $U\xgate_1U^{-1}$ and $U\zgate_1U^{-1}$ are single-qubit Pauli operators acting on the $k$-th qubit. Thus $U$ is $(1,k)$-non-entangling.  All operators $B_i$ defined above can be expressed using $\{\czgate,\cxgate^\downarrow\}$ gates, i.e. $B_i\,{\in}\, \calB_n$. Updating the stabilizer tableaux of $U$ by applying a single gate takes time $O(n)$. Since all updates performed above require $O(n)$ gates,  the overall runtime is $O(n^2)$.
\end{proof}

Suppose $U\,{\in}\, \calC_n$ is $(1,k_1)$-non-entangling.  Note that all operators $U \xgate_i U^{-1}$ and $U\zgate_i U^{-1}$ with $2\le i\le n$ act trivially on the qubit $k_1$.  Indeed, these operators must commute with $U \xgate_1 U^{-1}$ and $U\zgate_1 U^{-1}$ that generate the full Pauli group on the qubit $k_1$. Ignoring the action of $U$ on the first input qubit and restricting Pauli operators $U \xgate_i U^{-1}$ and $U\zgate_i U^{-1}$ with $2\le i\le n$ onto the subset of qubits $\irange{1}{n} \setminus \{k_1\}$ defines a Clifford operator $U'\,{\in}\, \calC_{n-1}$.  Applying \lem{algo3} to $U'$ gives $B_1,B_2 \in \calB_{n}$ and $k_2\in \irange{1}{n}{\setminus} \{k_1\}$ such that $B_1 U B_2$ is both $(1,k_1)$-non-entangling and $(2,k_2)$-non-entangling.  Proceeding inductively, one constructs $B_1,B_2\in \calB_n$ such that $B_1 U B_2$ is  $(j,k_j)$-non-entangling for each $j\,{\in}\, \irange{1}{n}$ and some $n$-tuple of distinct integers $(k_1,k_2,\ldots,k_n)$. By definition, it means that $B_1 U B_2$  is non-entangling.  Since there are $n$ disentangling steps, each taking time $O(n^2)$, the overall runtime is $O(n^3)$.  Applying  \lem{algo1}  with $U$ replaced by $B_1UB_2$ and multiplying the decomposition Eq.~(\ref{eq:nonentangling}) on the left and on the right by $B_1^{-1}$ and $B_2^{-1}$ respectively one gets
\be
\label{eq:disent4}
U= B_1^{-1} F_1 \left(\prod_{i=1}^n \hgate_i^{h_i} \right) S F_2 B_2^{-1}
\equiv  L \left(\prod_{i=1}^n \hgate_i^{h_i} \right) SR.
\ee
Recall that $F_1$ and $F_2$ are tensor products of single-qubit $\hgate$-free operators. In particular, $F_1,F_2\in \calB_n$. Since $B_1,B_2\in \calB_n$, one infers that $L=B_1^{-1}F_1\in \calB_n$ and $R=F_2B_2^{-1}\in \calB_n$.

By \lem{HHbar}, the operator $L$ in Eq.~(\ref{eq:disent4}) can be uniquely written as
\be
\label{eq:disent5}
L=KM, \quad K \in \calB_n(\bar{h},S), \quad M\in \calB_n(h,S).
\ee
Given $L$, how to compute $K$ and $M$?  To this end, we need multiplication rules of the $\hgate$-free group.
\begin{proposition}
\label{prop:mult}
Suppose $\Gamma_1,\,\Gamma_2$ are symmetric and $\Delta_1,\,\Delta_2$ are invertible $n\times n$ Boolean matrices; let $O_1,O_2\in \calP_n$ be any Pauli operators.  Then 
\be
\label{eq:mult0}
F(O_2,\Gamma_2,\Delta_2) \cdot F(O_1,\Gamma_1,\Delta_1)= F(O,\Gamma,\Delta),
\ee
where
\be
\label{eq:mult1}
\Delta = \Delta_2 \Delta_1 \quad \mbox{and} \quad \Gamma = \Gamma_1 \oplus \Delta_1^T \Gamma_2 \Delta_1.
\ee
Furthermore,
\be
\label{eq:mult2}
F(O_1,\Gamma_1,\Delta_1)^{-1} = F(O', (\Delta_1^{-1})^T \Gamma_1 \Delta_1^{-1}, \Delta_1^{-1}),
\ee
where $O,O'\in \calP_n$ are some Pauli operators.
\end{proposition}

\begin{proof}
We ignore the Pauli stages since they play no role; all Pauli gates can in fact be `commuted' to the side at any time through a process that does not change any of $\pgate$, $\czgate$, or $\cnotgate$ gates. This leaves the task of multiplying $F(Id,\Gamma_2,\Delta_2)$ and $F(Id,\Gamma_1,\Delta_1)$. For that, write both unitaries as layered circuits, -P$_1$-CZ$_1$-CX$_1$- and -P$_2$-CZ$_2$-CX$_2$-, and consider the result of circuit concatenation, -P$_1$-CZ$_1$-CX$_1$-P$_2$-CZ$_2$-CX$_2$-. Operator $\Gamma_2$ described by -P$_2$-CZ$_2$- experiences the application of phases to the linear functions of variables transformed by the matrix $\Delta_1$.  This means that the action of $\Gamma_2$ when applied to original primary circuit inputs is described by $\Delta_1^T \Gamma_2 \Delta_1$~\cite{maslov2018shorter}. This combines with the -P$_1$-CZ$_1$- stage by bitwise EXOR, since each $\czgate$ and $\pgate$ is self-inverse (subject to possibly factoring out a proper Pauli-Z), to obtain the second equality \eq{mult1}.  The reduced concatenated circuit expression now looks as -P-CZ-CX$_1$-CX$_2$-. The first equality \eq{mult1} follows directly (noting that the orders of circuit concatenation and matrix multiplications are inverted).  \eq{mult2} follows by verification.
\end{proof}

Write operators $K$, $L$, and $M$ in \eq{disent5} as
\[
L {=} F(O_1,\Gamma_1,\Delta_1), \,
K^{-1} {=} F(O_2,\Gamma_2,\Delta_2), \,
M {=} F(O, \Gamma, \Delta).
\]
where $\Gamma$ and $\Gamma_i$ are symmetric matrices, and $\Delta$ and $\Delta_i$ are lower-triangular unit-diagonal matrices, see Eq.~(\ref{eq:Borel1}).  For now, we will ignore the Pauli parts $O_i$ and $O$.  Clearly, $L\,{=}\,KM$ is equivalent to $K^{-1}L \,{=}\, M$.  By \prop{mult}, the latter is equivalent to Eq.~(\ref{eq:mult1}). Recall that $L$ has already been computed.  Thus the matrices $\Gamma_1$ and $\Delta_1$ are known and Eq.~(\ref{eq:mult1}) defines a linear system of equations with variables $\Gamma, \Delta$ and $\Gamma_2,\Delta_2$ parameterizing $M$ and $K^{-1}$.  Clearly, $K\,{\in}\, \calB_n(\bar{h},S)$ iff $K^{-1}\,{\in}\, \calB_n(\bar{h},S)$ since $\calB_n(\bar{h},S)$ is a group.  From  \lem{FHS} one infers that $K^{-1}\,{\in}\, \calB_n(\bar{h},S)$ iff $\Gamma_2, \Delta_2$ obey the rules C1-C5. Likewise, $M\,{\in}\, \calB_n(h,S)$ iff $\Gamma,\Delta$ obey  the rules C1-C5 with every bit of $h$ negated.  Importantly, the rules C1-C5 impose linear constraints  on matrix elements
of $\Gamma$ and $\Delta$.
 Combining Eq.~(\ref{eq:mult1}) with the rules C1-C5 one obtains a linear system with $O(n^2)$ variables $\Gamma,\Delta,\Gamma_2, \Delta_2$ and $O(n^2)$ equations. This linear system has a unique solution due to \lem{HHbar}. Such linear system can be solved in time $O(n^6)$ using the standard  algorithms. 
Finally, $K$ is computed from $K^{-1}$ using Eq.~(\ref{eq:mult2}).

Combining Eqs.~(\ref{eq:disent4},\ref{eq:disent5}) one arrives at
\[
U = KM W R=KW(W^{-1} M W)R, \quad W\equiv \left(\prod_{i=1}^n \hgate_i^{x_i} \right) S.
\]
Note that $W^{-1} M W\in \calB_n$ since $M\,{\in}\, \calB_n(h,S)$, see Eq.~(\ref{eq:bruhat003}).  Denoting $K'=(W^{-1} M W)R\in \calB_n$ one obtains $U=K W K'$ where $K\,{\in}\, \calB_n(\bar{h},S)$ and $K'\,{\in}\, \calB_n$.  This is the canonical form  stated in \thm{CliffordExact}. The Pauli part of $K'$ can be fixed by comparing the conjugated action of $U$ and $KWK'$ on Pauli operators.

\section{Generation of random Clifford operators}\label{sec:RandomClifford}

We next describe an algorithm for generating random uniformly distributed Clifford operators that utilizes the canonical form established in \thm{CliffordExact}.  Our algorithm runs in time $O(n^2)$, consumes $\log_2{|\calC_n|}$ random bits, and outputs  a Clifford operator sampled from the uniform  distribution on $\calC_n$. The Clifford operator is specified by its canonical form.  If needed, the canonical form can be converted to the stabilizer tableaux in time $O(n^\omega)$, where $\omega\,{\approx}\, 2.3727$ is the matrix multiplication exponent~\cite{williams2012multiplying}. Indeed, one can easily check that an operator $F\,{=}\,F(O,\Gamma,\Delta)$ defined in Eqs.~(\ref{eq:Hfree1},\ref{eq:Borel1}) has the stabilizer tableaux
\[
\left[ \ba{cc}
\Delta & 0 \\
\Gamma \Delta & (\Delta^{-1})^T \\
\ea
\right].
\]
Here the first $n$ columns represent Pauli operators $F\xgate_i F^{-1}$ (ignoring the phase) and the last $n$ columns represent $F\zgate_i F^{-1}$. Stabilizer tableau of the Hadamard stage and qubit permutation layers in the canonical form can be computed in time $O(n)$.  Finally, computing the inverse of the matrices $\Delta,\Delta'$ and multiplying the stabilizer tableau over all layers takes time $O(n^\omega)$. 

A simplified version of our algorithm, described in Appendix~\ref{app:A}, samples the uniform distribution on the group of invertible $n{\times} n$ binary matrices $\mathrm{GL}(n)$.  This algorithm has runtime $O(n^\omega)$ and consumes exactly $\log_2{|\mathrm{GL}(n)|}$ random bits. This improves upon the state-of-the-art algorithm due to Randall~\cite{randall1993efficient} which consumes $\log_2{|\mathrm{GL}(n)|}\,{+}\,O(1)$ random bits.

Given a bit string $h\,{\in}\, \{0,1\}^n$ and a permutation $S\,{\in}\, \calS_n$, define
\be
\label{eq:Wagain}
W=  \left( \prod_{i=1}^n \hgate_i^{h_i} \right) S.
\ee
Recall that the full Clifford group $\calC_n$ is a disjoint union of subsets
$\calB_n W \calB_n$.  We will need a normalized probability
distribution
\be
P_n(h,S) = \frac{|\calB_n W \calB_n|}{|\calC_n|}.
\ee
By definition, $P_n(h,S)$ is the fraction of $n$-qubit Clifford operators $U$ such that the canonical form of $U$ defined in \thm{CliffordExact} contains a layer of $\hgate$ gates labeled by $h$ and a qubit permutation $S$. We can write $|\calB_n W \calB_n| = |\calB_n|^2 \cdot |\calB_n(h,S)|^{-1}$, where $\calB_n(h,S)=\{B\in \calB_n \, {:} \, W^{-1} B W \in \calB_n\}$.  As we have already shown in Subsection~\ref{ssec:part1}, $|\calB_n(h,S)| = |\calB_n| 2^{-I_n(h,S)}$, where
\[
I_n(h,S) = n(n{-}1)/2 +|h| +  \sum_{\substack{1\le i<j\le n\\ S(i)<S(j) \\ }}\;  (-1)^{1+ h_i},
\]
see Eqs.~(\ref{eq:count1},\ref{eq:count2}).
Recalling that $|\calC_n|=|\calB_n|\prod_{j=1}^n(4^j-1)$ one arrives at
\be
\label{QMallows}
P_n(h,S)=\frac{2^{I_n(h,S)}}{\prod_{i=1}^n(4^i-1)}.
\ee
Interestingly, Eq.~(\ref{QMallows}) is a natural `symplectic' analogue of the well-known Mallows distribution~\cite{mallows1957non} on the symmetric group $\calS_n$ describing qubit permutations.  The latter is defined as
\begin{eqnarray*}
P_n(S) = \frac{2^{I_n(S)}}{\prod_{i=1}^n(2^i-1)}, \\
I_n(S) = \#\{ i,j\in \irange{1}{n} \, : \, i<j \;\; \mbox{and} \;\; S(i)>S(j)\}
\end{eqnarray*}
(a more general version of the Mallows distribution has probabilities $P_n(S)\sim q^{I_n(S)}$ for some $q\,{>}\,0$).  In particular, one can easily check that $I_n(0^n,S)=I_n(S)$.  The Mallows distribution is relevant in the context of ranking algorithms~\cite{lu2011learning}. It also plays a central role in our algorithm for sampling the uniform distribution on the group $\mathrm{GL}(n)$, see Appendix~\ref{app:A} for details.  Accordingly, we will refer to Eq.~(\ref{QMallows}) as a {\em quantum Mallows distribution}. Consider the following algorithm.

{\centering
\begin{minipage}{1.0\linewidth}
\begin{algorithm}[H]
	\caption{Generating $h,S$ per quantum Mallows distribution $P_n(h,S)$\label{Sampling_QMallows}}
	\begin{algorithmic}[1]
		\State{$A\gets \irange{1}{n}$}
		\For{$i=1$ to $n$}
		\State{$m\gets |A|$}
		\State{Sample $h_i\,{\in}\, \{0,1\}$ and   $k\in \irange{1}{m}$ from the probability vector  
		\[
		p(h_i,k) = \frac{2^{m -1 + h_i + (m-k)(-1)^{1+h_i}}}{4^{m}-1}.
		\]}				
		\State{Let $j$ be the $k$-th largest element of $A$}
		\State{$S(i)\gets j$}
		\State{$A\gets A\setminus \{j\}$}
		\EndFor
		\State{\textbf{return} $(h,S)$}
		\end{algorithmic}
\end{algorithm}
\end{minipage}
}

We included a Python implementation in Appendix~\ref{app:C}.

\begin{lemma}
\label{lem:qMallows}
Algorithm~\ref{Sampling_QMallows} outputs a bit string $h\,{\in}\, \{0,1\}^n$ and a permutation $S\,{\in}\, \calS_n$ sampled from the quantum Mallows distribution $P_n(h,S)$ defined in Eq.~(\ref{QMallows}).  The algorithm
can be implemented in time $\tilde{O}(n)$.  
\end{lemma}
\begin{proof}
Let us first check correctness of the algorithm. We use an induction in $n$. The base of induction is $n{=}1$. Note that $\calS_1{=}\{Id\}$ contains a single element, the identity.  Direct inspection shows that $I_1(0,Id) \,{=}\,0$ and $I_1(1,Id)\,{=}\,1$. Accordingly, $P_1(0,Id)\,{=}\,1/3$ and $P_1(1,Id)\,{=}\,2/3$.  The probability vector $p(h_1,1)$ sampled at Step~3 of the algorithm is $p(0,1)\,{=}\,1/3$ and $p(1,1)\,{=}\,2/3$. Thus the algorithm returns a sample from $P_1(h,Id)$,
as claimed. 

Next consider some fixed $n\,{>}\,1$. Let $k:= S(1)$ and $S'$ be the permutation obtained from $S$ by removing the first column and the $k$-th row of $S$ (when written as a linear invertible matrix). Let $h' = (h_2,h_3,\ldots,h_n)$.
Simple algebra gives 
\be
\label{eq:Ihs1}
I_n(h,S) = I_{n-1}(h',S') + n{-}1 + h_1 + (n{-}k)(-1)^{1+h_1}.
\ee
It follows that $P_n(h,S)\, {=}\, P_{n-1}(h',S') p(h_1,k)$, where $p(h_1,k)$ is the distribution sampled at Step~4 of Algorithm~\ref{Sampling_QMallows} at the first iteration of the  {\bf for} loop (with $i{=}1$).  Note that all subsequent iterations of the  {\bf for} loop (with $i\,{\ge}\, 2$) can be viewed as applying Algorithm~\ref{Sampling_QMallows} recursively to generate $h'\,{\in}\,\{0,1\}^{n-1}$ and $S'\,{\in}\, \calS_{n-1}$.  By the induction hypothesis, the probability of generating the pair $(h',S')$ is $P_{n-1}(h',S')$. Thus the probability of generating the pair $(h,S)$ is $P_n(h,S)$.

We claim that Step~3 of the algorithm can be implemented in time $\tilde{O}(1)$.  Indeed, consider some fixed iteration of the {\bf for} loop and let $m=|A|$.  Define a probability vector 
\begin{eqnarray*}
P:= [\underbrace{p(1,1), p(1,2), \ldots ,p(1,m)}_{h_i=1}, \\
\underbrace{p(0,m),p(0,m-1),\ldots,p(0,1)}_{h_i=0}].
\end{eqnarray*}
Simple algebra shows that
\[
P_a = \frac{2^{2m-a}}{4^m-1}, \quad a\in \irange{1}{2m}.
\]
A sample $a\in \irange{1}{2m}$ from the probability vector $P$ can be 
obtained as 
\[
a = 2m+1- \lceil \log_2{\left( r(4^m-1)+1\right)} \rceil,
\]
where $r\in [0,1]$ is a random uniform real variable. Here we used the fact that $P$ is the geometric series. Since the integer $a$ is represented using $O(\log{n})$ bits, the runtime is $\tilde{O}(1)$ per each iteration of the {\bf for} loop.
\end{proof}
Python language implementation of the quantum Mallows distribution algorithm can be found in Appendix~\ref{app:C}.  Next let us describe our algorithm for sampling the uniform distribution on the Clifford group.  We included a Python implementation in Appendix~\ref{app:D}.

\begin{figure}
{\centering
\begin{minipage}{1.0\linewidth}
\begin{algorithm}[H]
	\caption{Random $n$-qubit Clifford operator}\label{random_clifford_optimal}
	\begin{algorithmic}[1]	
	\State{Sample $h\,{\in}\, \{0,1\}^n$ and $S\,{\in}\, \calS_n$ from the quantum Mallows distribution $P_n(h,S)$}		  
	\State{Initialize $\Delta$ and $\Delta'$ by $n{\times}n$ identity matrices}
	\State{Initialize $\Gamma$ and $\Gamma'$ by $n{\times}n$ zero matrices}
	\For{$i=1$ to $n$}
	\State{Sample $b\,{\in}\, \{0,1\}$ from the uniform distribution and assign $\Gamma'_{i,i}=b$}
	\If{$h_i=1$}
	\State{Sample $b\,{\in}\, \{0,1\}$ from the uniform distribution and assign $\Gamma_{i,i}=b$}.
	\EndIf
	\EndFor
		\For{$j=1$ to $n$}
		\For{$i=j+1$ to $n$}
		\State{Sample $b\,{\in}\, \{0,1\}$ from the uniform distribution and assign $\Gamma'_{i,j}=\Gamma'_{j,i}=b$}
		\State{Sample $b\,{\in}\, \{0,1\}$ from the uniform distribution and assign $\Delta'_{i,j}=b$}
		\If{$h_i=1$ \myAND  $h_j=1$}  
		\State{Sample $b\,{\in}\, \{0,1\}$ from the uniform distribution and assign $\Gamma_{i,j}=\Gamma_{j,i}=b$}
		\EndIf
		\If{$h_i=1$ \myAND  $h_j=0$ \myAND $S(i)<S(j)$}  
		\State{Sample $b\,{\in}\, \{0,1\}$ from the uniform distribution and assign $\Gamma_{i,j}=\Gamma_{j,i}=b$}
		\EndIf
		\If{$h_i=0$ \myAND  $h_j=1$ \myAND $S(i)>S(j)$}  
		\State{Sample $b\,{\in}\, \{0,1\}$ from the uniform distribution and assign $\Gamma_{i,j}=\Gamma_{j,i}=b$}
		\EndIf
		\If{$h_i=0$ \myAND  $h_1=1$}  
		\State{Sample $b\,{\in}\, \{0,1\}$ from the uniform distribution and assign $\Delta_{i,j}=b$}
		\EndIf
		\If{$h_i=1$ \myAND  $h_j=1$ \myAND $S(i)>S(j)$}  
		\State{Sample $b\,{\in}\, \{0,1\}$ from the uniform distribution and assign $\Delta_{i,j}=b$}
		\EndIf
		\If{$h_i=0$ \myAND  $h_j=0$ \myAND $S(i)<S(j)$}  
		\State{Sample $b\,{\in}\, \{0,1\}$ from the uniform distribution and assign $\Delta_{i,j}=b$}
		\EndIf
	
		\EndFor
		\EndFor
		\State{Sample Pauli operator $O'\,{\in}\, \calP_n$ from the uniform distribution}
		\State{\textbf{return} $F(Id,\Gamma,\Delta) \cdot \left( \prod_{i=1}^n \hgate_i^{h_i} \right) S \cdot F(O',\Gamma',\Delta')$}
		\end{algorithmic}
\end{algorithm}
\end{minipage}
}
\end{figure}

We claim that this algorithm outputs a random uniformly distributed Clifford operator $U\,{\in}\, \calC_n$.  Indeed, Step~1 samples a random subset $\calB_n W\calB_n\subseteq \calC_n$ with $W$ defined in Eq.~(\ref{eq:Wagain})  with probability $|\calB_n W\calB_n|/|\calC_n|$.  All subsequent steps of the algorithm sample a random uniformly distributed operator $U\in \calB_n W\calB_n$.  Indeed, by \thm{CliffordExact}, such operator has the form $U = F(Id,\Gamma,\Delta) \cdot \left( \prod_{i=1}^n \hgate_i^{h_i} \right) S \cdot F(O',\Gamma',\Delta')$, where $\Gamma,\Gamma'$ are symmetric binary matrices, $\Delta,\Delta'$ are lower-triangular unit-diagonal matrices, such that $\Gamma,\Delta$ obey the rules C1-C5 of \thm{CliffordExact}.  These rules are enforced at Steps 14-31 of the above algorithm.

One may also use a simplified version of the algorithm where Steps~14-31 are skipped and $\Gamma_{i,j}$, $\Delta_{i,j}$ are sampled from the uniform distribution similar to Steps~12,13.  This simplified version still generates a random uniformly distributed Clifford operator but it is not optimal in terms of the number of random bits.

\section{Clifford circuit reduction}
\label{sec:er}

In this section we explore modifications of the canonical form established in \thm{CliffordExact} for reducing the gate count in Clifford circuits. 
 
We first focus on the following problem: given a Clifford operation defined by the circuit $C$, and applied to an unknown computational input state $\ket{x}$, construct as small of a Clifford circuit $D$ as possible such that the state $D(C(\ket{x}))$ is a computational basis state with possible phases.  Such a circuit $D$ removes the observable quantumness, including both entanglement and superposition, introduced by the circuit $C$.  It is beneficial to use the circuit $D$ in the last step before the measurement in randomized benchmarking protocols \cite{emerson2005scalable}, when it is shorter than $C^{-1}$.  Naturally, $D{:=}C^{-1}$ accomplishes the goal, however, a shorter circuit may exist.  In particular, as follows from \lem{lemCXCZCZ}, a Clifford circuit/state can be ``unentangled'' with the computation -H-CZ-H-P- at the effective two-qubit gate depth of $2n{+}2$ over Linear Nearest Neighbor architecture \cite[Theorem 6]{maslov2018shorter}.

\begin{lemma}\label{lem:ESRemoval}
Suppose $C$ is a Clifford operator such that the canonical form
of $C$ contains $k$ Hadamard gates. There exists a Clifford circuit $D$ with at most $nk-\frac{k(k+1)}{2}$ gates such that $CD \,{\in}\, \calF_n$.
\end{lemma}
To illustrate the advantage, suppose $n\,{=}\,5$. Then  $\max_{\{k\}}{nk-\frac{k(k+1)}{2}}=10$ (and in fact, $9$, if -CZ- counts are taken from \cite[Table 1]{maslov2018shorter}), whereas the maximal number of the two-qubit gates required by a $5$-qubit Clifford circuit is $12$.  
\begin{proof}
First, consider the canonical form $C = B_2{\cdot}SH {\cdot}B_1 = f_2\cdot\hgate^k\cdot f_1$, where $f_1\,{:=}\,B_1$, and $f_2$, that combines $B_2$, $S$, and a permutation of qubits that maps $H$ into $\hgate^k$, are Hadamard-free, and $\hgate^k$ is a layer of Hadamard gates on the top $k$ qubits. The circuit $f_1$ already contains a subset of gates from the canonical decomposition \eq{Borel1}, due to restrictions C1-C5 in \thm{CliffordExact}.  However, more gates can be moved into $f_2$ by observing that $f_2$ need not be restricted to contain only $\cnotgate^\downarrow$ gates. Once this is accomplished, the desired $D$ is defined as $f_1^{-1}\hgate^k$.  Clearly, $CD = f_2\cdot\hgate^k\cdot f_1f_1^{-1}\hgate^k = f_2$ is Hadamard-free, and the number of two-qubit gates in $D$ equals that in the reduced $f_1$.

The reduction of the number of gates in $f_1=F(Id,\Gamma,\Delta)$ takes four steps:
\begin{enumerate}
\item $\notgate$, $\zgate$ removal. Per \thm{CliffordExact}, the canonical decomposition \eq{Borel1} already contains no Pauli gates.
\item $\cnotgate$ removal. Write the $\cnotgate$ part of the computation as an invertible Boolean $n{\times}n$ matrix $\Delta$. The matrix $\Delta$ has the following block structure, 

\[
\Delta=
\begin{bmatrix}
\Delta_{k{\times}k} & 0 \\
\Delta_{k{\times}(n{-}k)} & \Delta_{(n{-}k){\times}(n{-}k)}
\end{bmatrix}.
\]
To implement the reversible linear transformation given by the matrix $\Delta$ we first find a sequence of gates 
\begin{eqnarray}\label{eq:abcdtoab0d}
g_1g_2...g_s: 
\begin{bmatrix}
\Delta_{k{\times}k} & 0 \\
\Delta_{k{\times}(n{-}k)} & \Delta_{(n{-}k){\times}(n{-}k)}
\end{bmatrix} 
\nonumber \\
\mapsto 
\begin{bmatrix}
\Delta_{k{\times}k}^\prime & 0 \\
0 & \Delta_{(n{-}k){\times}(n{-}k)}^\prime
\end{bmatrix},
\end{eqnarray}
and then a circuit $F^k_{\cnotgate}$ that finishes the implementation by fully diagonalizing the matrix $\Delta$.  The overall $\cnotgate$ circuit implementing $\Delta$ is obtained by the concatenation, $F^k_{\cnotgate}\cdot g_sg_{s-1}...g_1$.  Note that this means that we can keep only the $g_sg_{s-1}...g_1$ piece in $f_1$ and merge $F^k_{\cnotgate}$ with $f_2$.  The number of the $\cnotgate$ gates sufficient to preform the mapping in \eq{abcdtoab0d} is upper bounded by the number of non-zero matrix elements, which is at most $k(n{-}k)$. 

\item $\czgate$ removal. To reduce $\czgate$ gates, rewrite the three-stage circuit -CX-CZ-P- with $\cnotgate$-, $\czgate$-, and $\pgate$-gate stages as the three stage computation -CZ-P-CX- using phase polynomials \cite{maslov2018shorter}.  Observe that such transformation does not change the linear reversible circuit in the stage -CX-, and thus the number of gates in it remains minimized.  This transformation exposes $\czgate$ gates, all of which commute, on the left and allows merging them into the circuit $f_2$.  The only gates that remain in $f_1$ are those operating on the top $k$ qubits, of which there are at most $\frac{k(k-1)}{2}$.
\item $\pgate$ removal. Finally, to reduce $\pgate$ gates, move them through the remaining $\czgate$ gates (recalling that all diagonal gates commute) and cancel all but possibly top $k$ Phase gates.
\end{enumerate}  
The reduced $f_1$ contains no more than $k(n{-}k)$ $\cnotgate$ gates, no more than $\frac{k(k-1)}{2}$ $\czgate$ gates (and no more than $k$ $\pgate$ gates), proving lemma. 
\end{proof}

\begin{lemma}\label{lem:lemCXCZCZ}
An arbitrary Clifford circuit can be decomposed into stages -X-Z-P-CX-CZ-H-CZ-H-P-.
\end{lemma} 
\begin{proof}
We start with the decomposition $f_2 \cdot \hgate^k \cdot f_1$, where $f_2,f_1 \in \calF_n$, and perform a slightly different reduction of $f_1$ than the one described above. Specifically, we first write the layered expression as -X-Z-P-CX-CZ-H-C$_1$-CZ-P- (where a -P- layer may contain only first and third powers of $\pgate$).  As discussed in the previous proof, the stage -C$_1$- can be implemented by the $\cnotgate$ gates with controls on the top $k$ qubits and targets on the bottom $n{-}k$ qubits.  We next write -X-Z-P-CX-CZ-H-C$_1$-CZ-P- as -X-P-CX-CZ-H-CZ$_1$-C$_1$-P-, where the reduced stage -CZ$_1$- applies $\czgate$ gates to the top $k$ qubits.  Recall that changing the order of -CX- and -CZ- stages does not change the -CX- stage.  Write the stage -C$_1$- as -H$_1$-CZ$_2$-H$_1$- to obtain the decomposition of the form -X-Z-P-CX-CZ-H-CZ$_1$-H$_1$-CZ$_2$-H$_1$-P-. Observe that the Hadamard gates in -H$_1$- operate on the bottom $n{-}k$ qubits, and recall that $\czgate$ gates in -CZ$_1$- stage operate on the top $k$ qubits.  Thus, these two stages can be commuted to obtain -X-Z-P-CX-CZ-H-H$_1$-CZ$_1$-CZ$_2$-H$_1$-P-, that is equal to -X-Z-P-CX-CZ-H-CZ$_1$-CZ$_2$-H$_1$-P- once the Hadamard gate stages are merged, and further reduces to -X-Z-P-CX-CZ-H-CZ-H$_1$-P- by combining the neighboring -CZ- stages.  This obtains the desired expression.  
\end{proof}

We remark that the decomposition in \lem{lemCXCZCZ} has only three two-qubit gate stages.  It is similar to and refines the one reported in \cite{duncan2019graph}. This decomposition can be used to implement arbitrary Clifford operation in the two-qubit gate depth $9n$ in the Linear Nearest Neighbor architecture by rewriting it as -X-P-C$\widehat{\text{-CZ-}}$H$\widehat{\text{-CZ-}}$H-P-, where $\widehat{\text{-CZ-}}$ is -CZ- plus qubit order reversal, similar to how it is done in \cite[Corollary 7]{maslov2018shorter} and improving the previously known upper bound of $14n{-}4$.

Further reductions can be obtained by employing synthesis algorithms that exploit the structure better than the naive algorithms do, applying local optimizations, and looking up the implementation in the meet-in-the-middle style optimal synthesis approach for the elements of $\calF_n$.  Note that since the circuits implementing the reduced stage $f_1$ are small at the offset, chances are their optimal implementations have a less than average cost, and thus they may be found even by an incomplete meet-in-the-middle algorithm incapable of finding an optimal implementation of arbitrary elements of $\calF_n$.

\section{Quantum advantage for CNOT circuits}
\label{sec:CNOT}

In \sec{main} we studied efficient decompositions of the Clifford group into layers of Hadamard gates and $\hgate$-free circuits.  In our constructions, the $\hgate$-free circuits were expressed efficiently using combinations of layers with $\notgate$, $\pgate$, $\czgate$, and $\cnotgate$ gates (sometimes, $\zgate$ and $\swapgate$ gates were also used). One may ask if Hadamard gates can themselves be employed to obtain more efficient implementations of the $\hgate$-free transformations?  Surprisingly, the answer turns out to be ``yes''; this Section is devoted to the exploration of the efficient use of Hadamard gates in the implementation of $\hgate$-free operations.  Specifically, we focus on linear reversible circuits.  We show that so long as one is concerned with the entangling gate count (considering only $\cnotgate$ and $\czgate$ gates), a linear reversible function may be implemented more efficiently as a Clifford circuit as opposed to a circuit relying on the $\cnotgate$ gates, and prove two lemmas giving rise to two algorithms for optimizing the number of two-qubit gates in the $\cnotgate$ circuits.  Our result implies that quantum computations by Clifford circuits are more efficient than classical computations by reversible $\cnotgate$ circuits and computations of the $\calF_n$ by $\notgate$, $\pgate$, $\czgate$, and $\cnotgate$ circuits---see Example \ref{ex:1} below for explicit construction.

Given a bit string $a\,{\in}\, \FF_2^n$ let $H(a)$ be the product of Hadamards over all qubits  $j$ with $a_j{=}1$.  Suppose $W$ is a CNOT circuit on $n$ qubits and $a,b\in \FF_2^n$.  Here we derive necessary and sufficient conditions under which $H(b)WH(a) \in \calF_n$, i.e., it is $\hgate$-free.  First, define a linear subspace (vector columns)
\[
\calL(a) := \{ x\in \FF_2^n: \, \mathrm{Supp}(x)\subseteq \mathrm{Supp}(a)\}.
\]
\begin{lemma}\label{lem:lem3}
Suppose $W\,{=}\,\sum_x  |Ux\ra\la x|$ for some binary invertible matrix $U\,{\in}\, \mathrm{GL}(n)$.  The operator \linebreak $H(b)WH(a) \in \calF_n$ if and only if $U\cdot \calL(a) = \calL(b)$.
\end{lemma}
\begin{proof}
Define a state 
\[
|\psi\ra = H(b)WH(a)|0^n\ra.
\]
One can easily check that $H(b)WH(a) \in \calF_n$ iff $|\psi\ra$ is proportional to a basis vector. From $U\cdot \calL(a) = \calL(b)$ one gets $|\psi\ra =|0^n\ra$, that implies $H(b)WH(a) \in \calF_n$.

Conversely, suppose $H(b)WH(a)  \in \calF_n$, that is, $|\psi\ra = e^{k\pi i/2}|c\ra$ for some $c\,{\in}\, \FF_2^n$ and $k \,{\in}\, \{0,1,2,3\}$. Then
\begin{eqnarray*}
WH(a)|0^n\ra \sim H(b)|c\ra = H(b) X(c) |0^n\ra \\
=  Z(b\cap c) X(c\setminus b) H(b) |0^n\ra.
\end{eqnarray*}
Define a linear subspace $\calM = U\cdot \calL(a)$. Then
\[
\sum_{x\in \calM} |x\ra \sim Z(b\cap c) X(c\setminus b) \sum_{x\in \calL(b)} |x\ra.
\]
Let lhs and rhs be the left- and right-hand sides in this equation. The lhs has a non-zero amplitude for $|0^n\ra$ and so must the rhs. Thus $c\setminus b\in \calL(b)$ and we can ignore the Pauli-$X(c\setminus b)$. All amplitudes of the lhs have the same sign and so must amplitudes of the rhs. This is only possible if $b\cap c \in \calL(b)^\perp$ and we can ignore the Pauli-$Z(b\cap c)$.  Thus
\[
\sum_{x\in \calM} |x\ra \sim \sum_{x\in \calL(b)} |x\ra,
\]
which is only possible if $\calM=\calL(b)$, that is, $U\cdot \calL(a)=\calL(b)$.
\end{proof}

Based on \lem{lem3}, we formulate the following algorithm that optimizes the number of two-qubit gates in a $\cnotgate$ circuit. Given the $\cnotgate$ circuit $W$, check whether multiplying $W$ on the left and on the right by Hadamard gates on some subsets of qubits gives a transformation $\tilde{W}{:=}H(b)WH(a) \in \calF_n$.  Rely on a compiler for the elements of $\calF_n$, such as an optimal compiler for small $n$, to generate an efficient implementation of $\tilde{W}$.  If the number of two-qubit gates used by $\tilde{W}$ is smaller than that in the $W$, replace $W$ with $H(b)\tilde{W} H(a)$, leading to the two-qubit gate count reduction. 

Next let us state a sufficient condition under which $H(b)WH(a)=W$.
\begin{lemma}\label{lem:lem4}
Suppose $W=\sum_x  |Ux\ra\la x|$ for some binary invertible matrix $U\,{\in}\, \mathrm{GL}(n)$.  Suppose that $U\cdot \calL(a) = \calL(b)$ and $Ux=(U^{-1})^Tx$ for any $x\,{\in}\, \calL(a)$. Then, $H(b)WH(a)=W$.
\end{lemma}
\begin{proof}
Define $U^*:= (U^{-1})^T$ and $\tilde{W}:= H(b)WH(a)$. We use the following well-known identities:
\begin{eqnarray} \label{PauliProp}
WX(\gamma)W^{-1} = X(U\gamma) \,\; \mbox{and} \,\;
WZ(\gamma)W^{-1} = Z(U^*\gamma) \nonumber \\
\mbox{for all $\gamma\,{\in}\, \FF_2^n$}.
\end{eqnarray}
Write $X(\gamma)=X(\gamma_{\mathrm{in}})X(\gamma_{\mathrm{out}})$, where $\gamma_{\mathrm{in}}=\gamma \cap a$ and $\gamma_{\mathrm{out}}=\gamma \setminus a$. Then,
\begin{eqnarray}
\label{eq1}
\tilde{W} X(\gamma_{\mathrm{in}}) \tilde{W}^{-1}
= H(b) W Z(\gamma_{\mathrm{in}})W^{-1} H(b) \nonumber \\
=H(b) Z(U^*\gamma_{\mathrm{in}}) H(b) 
= X(U^*\gamma_{\mathrm{in}})=X(U\gamma_{\mathrm{in}}).
\end{eqnarray}
Here we used Eq.~(\ref{PauliProp}) and noted that $U^*\gamma_{\mathrm{in}}\in \calL(b)$.  Likewise,
\begin{eqnarray}
\label{eq2}
\tilde{W} X(\gamma_{\mathrm{out}}) \tilde{W}^{-1}
= H(b) W X(\gamma_{\mathrm{out}})W^{-1} H(b) \nonumber \\
=H(b) X(U\gamma_{\mathrm{out}}) H(b).
\end{eqnarray}
We claim that $U\gamma_{\mathrm{out}}$ and $b$ have non-overlapping supports. Indeed, pick any vector $\beta\,{\in}\, \calL(b)$. By assumption, $\beta =U^* \delta$ for some $\delta\,{\in}\, \calL(a)$. Thus,
\[
\beta^T U \gamma_{\mathrm{out}}=(U^*\delta)^T U \gamma_{\mathrm{out}}
= \delta^T U^{-1} U \gamma_{\mathrm{out}} 
= \delta^T \gamma_{\mathrm{out}}=0,
\]
since $\delta$ and $\gamma_{\mathrm{out}}$ have non-overlapping supports.  This shows that $U\gamma_{\mathrm{out}}$  is orthogonal to any vector in $\calL(b)$, that is,  $U\gamma_{\mathrm{out}}$ and $b$ have non-overlapping supports.  Thus $X(U\gamma_{\mathrm{out}})$ commutes with $H(b)$ and Eq.~(\ref{eq2}) gives
\be
\label{eq3}
\tilde{W} X(\gamma_{\mathrm{out}}) \tilde{W}^{-1}= X(U\gamma_{\mathrm{out}}).
\ee
Combining Eqs.~(\ref{eq1},\ref{eq3}) one arrives at
\[
\tilde{W} X(\gamma) \tilde{W}^{-1}= X(U\gamma_{\mathrm{in}})X(U\gamma_{\mathrm{out}})=X(U\gamma).
\]
Exactly the same argument applies to show that 
\[
\tilde{W}Z(\gamma)\tilde{W}^{-1} = Z(U^*\gamma).
\]
Thus $\tilde{W}=W$.
\end{proof}

The benefit of \lem{lem4} in comparison to \lem{lem3} is the $\tilde{W}$ constructed is equal to the original $W$, and thus the result of \lem{lem4} can be applied to subcircuits of a given $\cnotgate$ circuit. This allows the introduction of the Hadamard gates that can be used to turn some of the neighboring $\cnotgate$ gates into $\czgate$ gates, and then induce those $\czgate$ gates via a set of phases, thus requiring no two-qubit gates to implement the given $\czgate$'s.  We illustrate this optimization algorithm with an example.

\begin{example}\label{ex:1}
Consider $\cnotgate$-optimal reversible circuit:
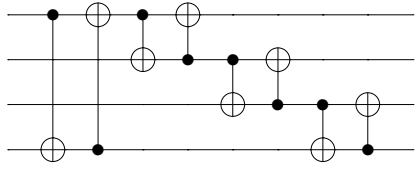
\[
\Qcircuit @C=0.7em @R=0.7em @!{
&\ctrl{3}	& \targ 	& \ctrl{1} 	&\targ		& \qw	 	& \qw 		& \qw 		& \qw 		& \qw \\
&\qw	 	& \qw		& \targ 	&\ctrl{-1} 	& \ctrl{1}	& \targ 	& \qw 		& \qw 		& \qw \\
&\qw	  	& \qw		& \qw 		&\qw 		& \targ 	& \ctrl{-1} & \ctrl{1} 	& \targ 	& \qw \\ 
&\targ		& \ctrl{-3} & \qw 		&\qw 		& \qw 		& \qw		& \targ 	& \ctrl{-1}	& \qw
}
\]
We established its optimality by breadth first search.  Notice that the subcircuit spanning gates $2$ through $8$ implements a transformation that satisfies the conditions of \lem{lem4} with $a=0001_2=1$ and $b=0010_2=2$. This allows to rewrite the above circuit in an equivalent form as
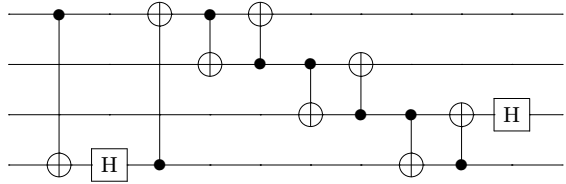
\[
\Qcircuit @C=0.5em @R=0.5em @!{
&\ctrl{3}	& \qw 			& \targ 	& \ctrl{1} 	&\targ		& \qw	 	& \qw 		& \qw 		& \qw 		& \qw 			& \qw \\
&\qw	 	& \qw 			& \qw		& \targ 	&\ctrl{-1} 	& \ctrl{1}	& \targ 	& \qw 		& \qw 		& \qw 			& \qw \\
&\qw	  	& \qw 			& \qw		& \qw 		&\qw 		& \targ 	& \ctrl{-1} & \ctrl{1} 	& \targ 	& \gate{\hgate}	& \qw \\ 
&\targ		& \gate{\hgate} & \ctrl{-3} & \qw 		&\qw 		& \qw 		& \qw		& \targ 	& \ctrl{-1}	& \qw 			& \qw
}
\]
We next commute leftmost $\cnotgate$ gate through leftmost Hadamard gate, by turning former into a $\czgate$ gate, to obtain
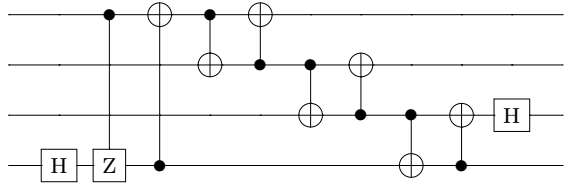
\[
\Qcircuit @C=0.5em @R=0.5em @!{
& \qw 			&\ctrl{3}		& \targ 	& \ctrl{1} 	&\targ		& \qw	 	& \qw 		& \qw 		& \qw 		& \qw 			& \qw \\
& \qw 			&\qw	 		& \qw		& \targ 	&\ctrl{-1} 	& \ctrl{1}	& \targ 	& \qw 		& \qw 		& \qw 			& \qw \\
& \qw 			&\qw	  		& \qw		& \qw 		&\qw 		& \targ 	& \ctrl{-1} & \ctrl{1} 	& \targ 	& \gate{\hgate}	& \qw \\ 
& \gate{\hgate}	&\gate{\zgate}	& \ctrl{-3} & \qw 		&\qw 		& \qw 		& \qw		& \targ 	& \ctrl{-1}	& \qw 			& \qw
}
\]
and notice that in the phase polynomial formalism the $\czgate(x,y)$ can be implemented as a set of Phase gates $\pgate(x)$, $\pgate(y)$, $\pgate^\dagger(x{\oplus}y)$. All three linear functions, $x$, $y$, and $x{\oplus}y$, already exist in the linear reversible circuit.  Thus, we can induce the desired $\czgate$ by inserting Phase gates (for simplicity, we insert Phase gates at the first suitable occurrence) to obtain the following optimized circuit:
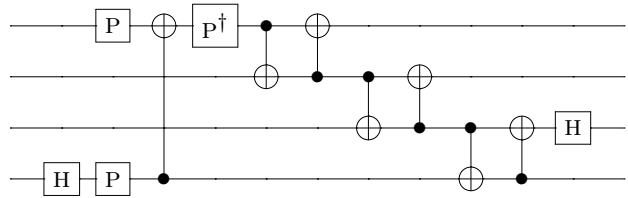
\[
\Qcircuit @C=0.2em @R=0.2em @!{
& \qw 		&\gate{\pgate}	& \targ 	& \gate{\pgate^\dagger}	& \ctrl{1} 	&\targ		& \qw	 	& \qw 		& \qw 		& \qw 		& \qw 			& \qw \\
& \qw 			&\qw	 		& \qw		& \qw	& \targ 	&\ctrl{-1} 	& \ctrl{1}	& \targ 	& \qw 		& \qw 		& \qw 			& \qw \\
& \qw 			&\qw	  		& \qw		& \qw	& \qw 		&\qw 		& \targ 	& \ctrl{-1} & \ctrl{1} 	& \targ 	& \gate{\hgate}	& \qw \\ 
& \gate{\hgate}	&\gate{\pgate}	& \ctrl{-3} & \qw	& \qw 		&\qw 		& \qw 		& \qw		& \targ 	& \ctrl{-1}	& \qw 			& \qw
}
\]
This circuit contains only $7$ two-qubit gates, thus beating the best possible $\cnotgate$ circuit while faithfully implementing the underlying linear reversible transformation.  Note that this circuit can also be obtained using the algorithm resulting from \lem{lem3} with an access to proper $\calF(4)$ compiler, such as the meet-in-the-middle optimal compiler.   
\end{example}

Above we illustrated the advantage quantum Clifford circuits offer over linear reversible $\cnotgate$ circuits.  Next, we study the limitations by such an advantage.  
\begin{lemma}
Quantum advantage by Clifford circuits over reversible $\cnotgate$ circuits cannot be more than by a constant factor.
\end{lemma}
\begin{proof}
Suppose $U$ is a classical reversible linear function on $n$ bits that can be implemented by a Clifford circuit $W$ with $g$ $\cnotgate$, and $G$ total gates from the library $\{\pgate, \hgate, \cnotgate \}$.  Trivially, $G \,{\in}\, O(g)$, as we can ignore all qubits affected by the single-qubit gates only.  We next show that there exists a $\cnotgate$ circuit composed of $O(g)$ $\cnotgate$ gates, spanning $2n$ bits, and implementing $U$, and thus quantum advantage by Clifford circuits cannot be more than by a constant factor.  Indeed, consider the operation $U$ written in the tableaux form \cite{aaronson2004improved}.  Each gate, $\pgate, \hgate$ and $\cnotgate$ is equivalent to a simple linear transformation of the $2n{\times}2n$ tableaux.  Specifically \cite{aaronson2004improved}, $\pgate$ corresponds to a column addition (one $\cnotgate$ gate over a linear reversible computation with $2n$ bits),  $\hgate$ corresponds to the column swap (three $\cnotgate$ gates over a linear reversible computation with $2n$ bits), and $\cnotgate$ corresponds to the addition of a pair of columns (two $\cnotgate$ gates over a linear reversible computation with $2n$ bits).  Thus, the total number of $\cnotgate$ gates in a $2n{\times}2n$ linear reversible computation implementing the tableau does not exceed $3G = O(g)$.  We conclude the proof by noting that the top left $n{\times}n$ block of the tableau coincides with the original linear transformation $U$. 
\end{proof}

We already demonstrated that classical reversible linear functions can be implemented with fewer $\cnotgate$ gates by making use of single-qubit Clifford gates.  More precisely, suppose $U$ is a classical reversible linear function.  Let $g=g(U)$ and $G=G(U)$ be the minimum number of $\cnotgate$s required to implement $U$ using gate libraries $\{\cnotgate \}$ and $\{\pgate, \hgate, \cnotgate \}$ respectively.  We have shown that $\Omega(g)\le G\le g$. Furthermore, $G<g$ in certain cases.

Let us say that a quantum circuit $W$ implements $U$ {\it modulo phases} if $W|x\ra = e^{i f(x)} |U(x)\ra$ for any basis state $x$, where $f(x)$ is an arbitrary function.  Note that the extra phases can be removed simply by measuring each qubit in the $0,1$ basis.  Let $G^*=G^*(U)$ be the minimum number of $\cnotgate$s required to implement $U$ modulo phases using the gate library $\{\pgate, \hgate, \cnotgate\}$.  By definition, $G^*\le G$. A straightforward example when $G^*<G$ is provided by the SWAP gate.  A simple calculation shows that
\be
\label{SWAPexample}
\swapgate\cdot  \czgate =  \hgate^{\otimes 2} \cdot \czgate\cdot  \hgate^{\otimes 2}\cdot \czgate \cdot  \hgate^{\otimes 2}.
\ee
The right-hand side can be easily transformed into a circuit with two $\cnotgate$s and two Hadamard gates while the left-hand side implements SWAP gate modulo phases showing that $G^*(\swapgate)\le 2$.  At the same time, it is well-known that $G(\swapgate)=3$. The following lemma provides a multi-qubit generalization of this example.
\begin{lemma}
Consider a classical reversible linear function $U$ on $n$ bits such that
$U(x)\,{=}\,Ax$ for some symmetric matrix $A\,{\in}\, \mathrm{GL}(n)$. 
Then 
\be
\label{G*bound}
G^*(U)\le \sum_{1\le i<j\le n} A_{i,j} + (A^{-1})_{i,j}.
\ee
\end{lemma}
\begin{proof}
Define Clifford operators $W=F(I,0,A)$, $D_+=F(I,A,0)$, and $D_-=F(I,A^{-1},0)$.  Here we used the $F(O,\Gamma,\Delta)$ notation introduced in \sec{main}.  Given a Clifford operator $C$, let $\tau(C)$ be the stabilizer tableaux of $C$.  A simple calculation gives
\begin{eqnarray*}
\tau(W)= \left[ \begin{array}{cc}
A & 0 \\
0 & A^{-1} \\
\end{array} \right], \;
\tau(D_+) = \left[ \begin{array}{cc}
I & 0 \\
A & I \\
\end{array} \right], \\
\tau(D_-) = \left[ \begin{array}{cc}
I & 0 \\
A^{-1} & I \\
\end{array} \right], \;
\tau(\hgate^{\otimes n}) = 
 \left[ \begin{array}{cc}
0 & I \\
I & 0 \\
\end{array} \right],
\end{eqnarray*}
and
\[
\tau(W)=\tau(\hgate^{\otimes n})\cdot
\tau(D_+)\cdot \tau(\hgate^{\otimes n})\cdot
\tau(D_-)\cdot \tau(\hgate^{\otimes n})\cdot \tau(D_+).
\]
Thus
\[
W = Q\cdot \hgate^{\otimes n} \cdot D_+ \cdot \hgate^{\otimes n} \cdot D_- \cdot \hgate^{\otimes n} \cdot D_+,
\]
where $Q$ is some Pauli operator. Equivalently,
\be
\label{WD+}
W\cdot D_+^{-1} =Q\cdot \hgate^{\otimes n} \cdot D_+ \cdot \hgate^{\otimes n} \cdot D_- \cdot \hgate^{\otimes n}.
\ee
The right-hand side can be implemented using 
\[
\sum_{1\le i<j\le n} A_{i,j} + (A^{-1})_{i,j}
\]
$\cnotgate$ gates. The left-hand side of Eq.~(\ref{WD+}) implements $U$ modulo phases.
\end{proof}
For example, if $U\,{=}\,\swapgate$ then $A=A^{-1}=\left[\begin{array}{cc} 0 & 1\\1 & 0 \\ \end{array}\right]$ and the bound in Eq.~(\ref{G*bound}) gives $G^*(\swapgate){\le}2$.

The following revisits some of the results discussed above, and formulates a new conjecture.
\begin{conj}
Define $\textup{Cost}$ to be the number of $\cnotgate$ gates in a given circuit.  For a linear reversible function $f$, consider the following optimal circuits, with respect to the above $\textup{Cost}$ metric:
\begin{itemize}
\item $A$ is an optimal $\{\cnotgate\}$ circuit implementing $f$;
\item $B$ is an optimal $\{\hgate, \pgate, \cnotgate\}$ circuit implementing $f$;
\item $C$ is an optimal $\{\hgate, \pgate, \cnotgate\}$ circuit implementing $f$ up to phases (i.e., $C: \ket{x} \mapsto i^k\ket{f(x)}$ for complex-valued $i$ and $k \,{\in}\, \{0,1,2,3\}$);
\item $D$ is an optimal $\{\cnotgate\}$ implementation of $f$ up to the SWAPping of qubits.
\end{itemize}
Then,
$$\textup{Cost}(A) \geq \textup{Cost}(B) \geq \textup{Cost}(C) \geq  \textup{Cost}(D).$$ 
\end{conj}
First two inequalities are straightforward.  Less trivial is to establish that these are proper inequalities, but as shown earlier each is indeed a proper inequality.  We conjecture that optimizations of linear reversible $\{\cnotgate\}$ circuits by considering Clifford circuits and by admitting implementations up to a relative phase do not exceed optimizations from allowing to implement the $\{\cnotgate\}$ circuit up to a permutation/SWAPping of the qubits.

\section{Conclusion}
In this paper, we studied the structure of the Clifford group and developed exactly optimal parametrization of the Clifford group by layered quantum circuits.  We leveraged this decomposition to demonstrate an efficient $O(n^2)$ algorithm to draw a random uniformly distributed $n$-qubit Clifford unitary, showed how to construct a small circuit that removes the entanglement from Clifford circuits, and implemented Clifford circuits in the Linear Nearest Neighbor architecture in the two-qubit depth of only $9n$.  Hadamard-free operations play a key role in our decomposition.  We studied Clifford circuits for the Hadamard-free set and showed the advantage by implementations with Hadamard gates; in particular, the smallest linear reversible circuit admitting advantage by Clifford gates has an optimal number of $8$ $\cnotgate$ gates (as a $\{\cnotgate\}$ circuit), but it can be implemented with only $7$ entangling gates as a Clifford circuit.

\section*{Acknowledgements}
Authors thank Dr. Jay Gambetta and Dr. Kevin Krsulich from IBM Thomas J. Watson Research Center for their helpful discussions. 
SB acknowledges the support of the IBM Research Frontiers Institute.


%
%
%
%
%
%

\appendix
\section{Generation of a random uniformly distributed Boolean invertible matrix}
\label{app:A}

In this section we consider the task of generating a random uniform
element of the binary general linear group
\[
\mathrm{GL}(n) =\{ M\in \FF_2^{n\times n}: \det{(M)}={1 \pmod 2} \}.
\]
It is well-known that  picking a random uniform matrix $M\,{\in}\, \FF_2^{n\times n}$ and testing whether $M$ is invertible produces a random uniform element of $\mathrm{GL}(n)$ with the success probability $\alpha_n:=\prod_{j=1}^n \frac{2^j-1}{2^j} \approx 0.2887\ldots$. The success probability can be amplified to $1\,{-}\,(1{-}\alpha_n)^m$ by repeating the protocol $m$ times.

In practice, testing the invertibility takes time $O(n^3)$, which can be expensive for large $n$. A more efficient algorithm avoiding the invertibility test was proposed by Dana Randall~\cite{randall1993efficient}. This algorithm has the runtime $O(n^2) + M(n)$, where $M(n)$ is the time it takes to multiply a pair of binary matrices of size $n{\times}n$.  Here we describe a closely related algorithm for sampling $\mathrm{GL}(n)$ that is based on the Bruhat decomposition. The algorithm has the runtime $O(n^2)+M(n)$ with a small constant coefficient and consumes exactly $\log_2{|\mathrm{GL}(n)|}$ random bits, which is optimal. In contrast,  Randall's algorithm is asymptotically tight in the random bit requirement, but not exactly optimal; indeed, it consumes $\log_2{|\mathrm{GL}(n)|}+O(1)$ random bits.

Let $\calS_n$ be the symmetric group over the set of integers $\irange{1}{n}$.
We consider permutations $S\,{\in}\, \calS_n$ as bijective maps $S: \irange{1}{n}\to \irange{1}{n}$ and write $S(i)$ for the image of $i$ under the action of $S$. Recall that the inversion number $I_n(S)$ of a permutation $S\,{\in}\, \calS_n$ is defined as the number of integer pairs $(i,j)$ such that $1\le i<j\le n$ and $S(i)>S(j)$.  By definition, $I_n(S)$ takes values between $0$ and $n(n-1)/2$.  The Mallows measure~\cite{mallows1957non} is the  probability distribution on $\calS_n$ defined as
\be
\label{eq:Mallows}
P_n(S):=\frac{2^{I_n(S)}}{\prod_{j=1}^n (2^j-1)}
\ee
(a more general version of the Mallows distribution has probabilities $P_n(S)\sim q^{I_n(S)}$
for some $q>0$).
Below we describe a  nearly-linear time algorithm for sampling the distribution $P_n(S)$.  
It is used as a subroutine in the algorithm for sampling the uniform distribution on the
group $\mathrm{GL}(n)$.
We identify a permutation $S{\,{\in}\,}\calS_n$ with the permutation matrix $S{\,{\in}\,}\FF_2^{n\times n}$ such that $S_{j,i}\,{=}\,1$ if $S(i)\,{=}\,j$, and else $S_{j,i}\,{=}\,0$. We use the permutation $S$ and the matrix $S$ corresponding to it interchangeably, depending on the context.  

{\centering
	\begin{minipage}{1.0\linewidth}
		\begin{algorithm}[H]
			\caption{Generating $S$ per Mallows distribution $P_n(S)$\label{Sampling_Mallows}}
			\begin{algorithmic}[1]
				\State{$A\gets \irange{1}{n}$}
				\For{$i=1$ to $n$}
				\State{$m\gets |A|$}
				\State{Sample $k\in \irange{1}{m}$ from the probability vector  $p_k = 2^{k-1}/(2^{|A|}-1)$}		
				\State{Let $j$ be the $k$-th largest element of $A$}
				\State{$S(i) \gets j$ }
				\State{$A\gets A\setminus \{j\}$}
				\EndFor
				\State{\textbf{return} $S$}
			\end{algorithmic}
		\end{algorithm}
	\end{minipage}
}

\begin{lemma}
	\label{lem:Mallows}
	Algorithm~\ref{Sampling_Mallows} outputs a permutation matrix $S$ sampled from the distribution $P_n(S)$ defined in \eq{Mallows}.  The algorithm  can be implemented in time $\tilde{O}(n)$.
\end{lemma}

\begin{proof}
	Let us first check correctness of the algorithm. We use an induction in $n$. The base of induction is $n{=}1$,
	is trivial. Suppose we proved the lemma for permutations of size $n{-}1$.  Consider a fixed permutation $S\,{\in}\,\calS_n$. Let $k=S(1)$.  The probability that Algorithm~\ref{Sampling_Mallows} picks $k$ at the first iteration of the {\bf for} loop (with $i=1$) is $p_k = 2^{k-1}/(2^n-1)$.  Let $S'\,{\in}\, \calS_{n-1}$ be a permutation corresponding to the matrix obtained by removing the first column and the $k$-th row of $S$. Simple algebra shows that $I_n(S) = k{-}1 + I_{n-1}(S')$. Note that all subsequent iterations of the  {\bf for} loop (with $i\ge 2$) can be viewed as applying Algorithm~\ref{Sampling_Mallows} recursively to generate $S'$ for $S'\,{\in}\, \calS_{n-1}$. By the induction hypothesis, the probability of generating $S'$ is $P_{n-1}(S')$. Thus the probability of generating $S$ is given by 
	\begin{eqnarray*}
	p_k \cdot P_{n-1}(S') 
	=\frac{2^{k-1}}{2^n-1} \cdot \frac{2^{I_{n-1}(S')}}{\prod_{j=1}^{n-1} (2^j-1)} \\
	=\frac{2^{I_n(S)}}{\prod_{j=1}^n (2^j-1)}
	=P_n(S).
	\end{eqnarray*}
	Consider some fixed iteration of the  {\bf for} loop and let $m=|A|$.
	A sample $k\in \irange{1}{m}$ from the probability vector  $p$ can be 
	obtained as 
	\[
	k = \lceil \log_2{\left( r(2^m-1)+1\right)} \rceil
	\]
	where $r\in [0,1]$ is a random
	uniform real variable. Here we used the fact that $p$ is the geometric series.
	Since the integer $k$ is represented using $\log_2{n}$ bits,
	the runtime is $\tilde{O}(1)$ per each iteration of the {\bf for} loop.
\end{proof}

Our algorithm for sampling the uniform distribution on $\mathrm{GL}(n)$ is stated below.

{\centering
	\begin{minipage}{1.0\linewidth}
		\begin{algorithm}[H]
			\caption{Random Invertible $n{\times}n$ Matrix }\label{random_invertible_optimal}
			\begin{algorithmic}[1]	
				\State{Sample a permutation $S$ for $S\,{\in}\, \calS_n$ from the Mallows distribution $P_n(S)$.}		\State{Initialize $L$ and $R$ by $n{\times}n$ identity matrices}
				\For{$j=1$ to $n$}
				\For{$i=(j+1)$ to $n$}
				\State{Sample $b\,{\in}\, \{0,1\}$ from the uniform distribution and assign $L_{i,j}=b$}
				\If{$S(i)<S(j)$}
				\State{Sample $b\,{\in}\, \{0,1\}$ from the uniform distribution and assign $R_{i,j} = b$}
				\EndIf
				\EndFor
				\EndFor
				\State{\textbf{return} $LSR$}
			\end{algorithmic}
		\end{algorithm}
	\end{minipage}
}

Step~1 takes time $\tilde{O}(n)$, see \lem{Mallows}.  Steps~2-10 take time $O(n^2)$ while step~11 takes time $O(n^2) + M(n)$. Thus the algorithm has runtime $O(n^2) + M(n)$, as claimed.

In the rest of this section we prove the correctness and the optimality of Algorithm~\ref{random_invertible_optimal} in the random bit count. Let $\calL_n$ be the set of all $n{\times}n$ lower-triangular unit-diagonal matrices,
\be
\label{Ln}
\calL_n := \{ M {\in} \FF_2^{n\times n}:  \mathrm{det}(M){=}1, \; \mbox{and} \; M_{i,j}{=}0 \; \mbox{for all $i{<}j$}\}.
\ee
In a certain sense, $\calL_n$ is an analogue of the Borel group $\calB_n$
that we used to establish canonical form of Clifford operators. 
Given a matrix $S\,{\in}\, \calS_n$, define the set $\calL_n S \calL_n:= \{ ASB: \, A,B\in \calL_n\}$.
Note that $\calL_n$ and $\calS_n$ are subgroups of $\mathrm{GL}(n)$.
We will need the following three lemmas.
\begin{lemma}[\bf Bruhat decomposition]
	\label{lem:bruhatGL}
	The group $\mathrm{GL}(n)$ is a disjoint union of sets $\calL_n S  \calL_n$,
	$$\mathrm{GL}(n) = \bigsqcup_{S\in \calS_n} \calL_n S  \calL_n.$$
\end{lemma}
\begin{proof}
	Multiplying a matrix $M \,{\in}\, \mathrm{GL}(n)$ on the left by a suitable element of $\calL_n$ one can add the $i$-th row of $M$ to the $j$-th row of $M$ modulo two for any $i\,{<}\,j$. This enables a version of Gaussian elimination that clears a column of $M$ downwards leaving the topmost nonzero element of the column untouched. Likewise, multiplying $M$ on the right by a suitable element of $\calL_n$ one can add the $j$-th column of $M$ to the $i$-th column of $M$ modulo two for any $i\,{<}\,j$. This enables a version of Gaussian elimination that clears a row of $M$ leftwards leaving the rightmost nonzero element of the row untouched.  Combined, a sequence of such operations can transform any invertible binary matrix to a permutation matrix. This shows that for any $M\,{\in}\, \mathrm{GL}(n)$ there exist $A,B \in \calL_n$ such that $AMB\,{\in}\, \calS_n$.  Thus $M\in A^{-1}\calS_n B^{-1}$. Since $\calL_n$ is a group, we conclude that $M\in \calL_n \calS_n \calL_n$, that is, $\mathrm{GL}(n)$ is a union of the subsets  $\calL_n S \calL_n$.
	
	It remains to check that the subsets $\calL_n S  \calL_n$ are pairwise disjoint.  Assume the contrary. Then $S\in \calL_nT \calL_n$ for some distinct permutation matrices $S,T\in \calS_n$. Equivalently, $S^{-1} L T\in \calL_n$ for some $L\,{\in}\, \calL_n$. Define $U:=S^{-1} T$. Since $U$ is a non-identity permutation, there exists an integer $i$ such that $i\,{>}\,U(i)$. Let $j\,{=}\,U(i)$. Then $S(j)\,{=}\,T(i)$ and $j\,{<}\,i$. This gives  $0=(S^{-1}LT)_{j,i}=L_{S(j),T(i)}=1$ since, by assumption, both $S^{-1} LT$ and $L$ are elements of $\calL_n$. We arrived at a contradiction. Thus $S\in \calL_n T \calL_n$ is possible only if $S\,{=}\,T$.
\end{proof}

\begin{lemma}
	\label{lem:double_coset_size}
	For all $S\,{\in}\, \calS_n$ the following equality holds
	\be
	\label{eq:double_coset_size}
	\frac{ |\calL_n S \calL_n|}{|\mathrm{GL}(n)|} = P_n(S),
	\ee
	where $P_n(S)$ is the Mallows distribution defined in \eq{Mallows}.
\end{lemma}
\begin{proof}
	Subsets $\calL_n$ and $S^{-1}  \calL_n S$ are the groups of size $|\calL_n|$. Thus 
	\be
	|\calL_n  S\calL_n| = 
	|(S^{-1}\calL_n  S) \calL_n|= 
	\frac{|\calL_n|^2}{| (S^{-1} \calL_n S) \cap \calL_n|}.
	\ee
	The set $(S^{-1} \calL_n S)\cap \calL_n$ includes all  unit-diagonal matrices $M$ such that $M_{i,j}=0$ for all pairs $(i,j)$ satisfying $i<j$ or $S(i)<S(j)$.  The number of such pairs is $n(n{-}1)/2 -I(S)$. Therefore $| (S^{-1} \calL_n S) \cap \calL_n|=2^{n(n-1)/2 - I(S)}$. Noting  that $|\calL_n|=2^{n(n-1)/2}$ and $|\mathrm{GL}(n)|=2^{n(n-1)/2} \prod_{i=1}^n (2^i-1)$ gives \eq{double_coset_size}.
\end{proof}

\lem{bruhatGL} and \lem{double_coset_size} immediately imply that a random uniform element of $\mathrm{GL}(n)$ can be generated as the matrix product $LSR$, where $L,R\in \calL_n$ are sampled from the uniform distribution and $S\,{\in}\, \calS_n$ is sampled from $P_n(S)$. Such simplified algorithm, however, is not optimal in terms of the number of random bits.  Indeed, it may happen that different choices of $L$ and $R$ give the same product $LSR$ (consider, as an example, the case $S\,{=}\,I$).  This is the reason why Algorithm~\ref{random_invertible_optimal} samples $R$ from a certain non-uniform distribution depending on $S$. Given a permutation $S\,{\in}\, \calS_n$ let $\calL_n(S)$ be the set of all matrices $R\,{\in}\, \calL_n$ such that $R_{i,j}\,{=}\,0$ for all pairs $i\,{>}\,j$ with $S(i)\,{>}\,S(j)$. Note that
\be
\label{eq:LnS}
|\calL_n(S)|=2^{I_n(S)}.
\ee
$\calL_n(S)$ can be viewed as a classical analogue of the group $\calB_n(h,S)$ that we used to establish canonical form of Clifford operators.  We next refine the statement of \lem{bruhatGL} obtaining an analogue of \thm{CliffordExact} for the group $\mathrm{GL}(n)$. 
\begin{lemma}[\bf Canonical Form]
	\label{lem:exact}
	Any element of $\mathrm{GL}(n)$ can be uniquely represented as $LSR$ for some $L\,{\in}\, \calL_n$, $S\,{\in}\, \calS_n$, and $R\,{\in}\, \calL_n(S)$.
\end{lemma}
\begin{proof}
	Let us first check that the representation stated in the lemma is unique. Suppose $LSR=L'S'R'$ for some $L,L'\in \calL_n$, some $S,S'\in \calS_n$, and some $R,R'\in \calL_n(S)$. By \lem{bruhatGL}, $S\,{=}\,S'$.  Denoting $L''=(L')^{-1} L$ one gets $S^{-1} L''S=R'R^{-1}$.  Thus
	\be
	\label{eq:tricky1}
	K:= R'R^{-1} \in (S^{-1} \calL_n S) \cap \calL_n.
	\ee
	Note that $K\in (S^{-1} \calL_n S) \cap \calL_n$ iff $K$ has unit diagonal and 
	\be
	\label{eq:tricky2}
	K_{i,j}=0 \quad \mbox{if $i<j$ or $S(i)<S(j)$}.
	\ee
	Let $T\,{\in}\, \calS_n$ be the reflection such that $T(i)=n-i+1$. For any permutation $S\,{\in}\, \calS_n$ let $\bar{S}\,{=}\,TS$. Note that $S(i)\,{<}\,S(j)$ iff $\bar{S}(i)\,{>}\,\bar{S}(j)$. From \eq{tricky1} and \eq{tricky2} and the definition of $\calL_n(S)$ one gets 
	\be
	\label{eq:SSbar}
	\calL_n(S)= (\bar{S}^{-1} \calL_n \bar{S}) \cap \calL_n.
	\ee
	This shows that $\calL_n(S)$ is a group. Furthermore,
	\be
	\label{eq:interSSbar}
	\calL_n(S) \cap \calL_n(\bar{S})=\{I\}.
	\ee
	From \eq{tricky1} \eq{SSbar} we have  $K\,{=}\,R'R^{-1}\in \calL_n(\bar{S})$ while $R,R'\in \calL_n(S)$. From \eq{interSSbar} and the fact that $\calL_n(S)$ is a group one infers that $K\,{=}\,I$, that is, $R\,{=}\,R'$. From $S\,{=}\,S'$ and $R\,{=}\,R'$, one gets $S\,{=}\,S'$. Thus the decomposition claimed in the lemma is unique, whenever it exists.
	
	By counting, any element of $\mathrm{GL}(n)$ admits such decomposition. Indeed, using \eq{LnS} one concludes that  the number of distinct triples $(L,S,R)$ with $L \,{\in}\, \calL_n$, $S\,{\in}\, \calS_n$, and $R\,{\in}\, \calL_n(S)$ is 
	\begin{eqnarray}\label{eq:counting}
	|\calL_n| \sum_{S\in \calS_n} |\calL_n(S)| = 2^{n(n-1)/2} \sum_{S\in \calS_n}  2^{I_n(S)} \nonumber \\
	=2^{n(n-1)/2}  \prod_{j=1}^n (2^j-1) = |\mathrm{GL}(n)|.
	\end{eqnarray}
	Here the second equality follows from the normalization of the Mallows distribution, see \eq{Mallows}.
\end{proof}

Combining \lem{exact} and \eq{LnS} one infers that the fraction of elements $M\,{\in}\, \mathrm{GL}(n)$ representable as $M\,{=}\,LSR$ with a given $S\,{\in}\, \calS_n$ is given by the Mallows measure $P_n(S)$.  Thus a random uniform $M\,{\in}\, \mathrm{GL}(n)$ can be obtained by sampling $S$ from $P_n(S)$, sampling $L$ uniformly from $\calL_n$, sampling $R$ uniformly from $\calL_n(S)$, and returning $LSR$.  This is described in  Algorithm~\ref{random_invertible_optimal}.  The number of random bits consumed by the algorithm is exactly $\log_2{|\mathrm{GL}(n)|}$ since it is based on the exact parameterization of the group.

\section{Proof of the existence of exact parametrization of the Clifford group}
\label{app:B}

In this appendix, we report a shorter proof of the existence of exact parametrization of the Clifford group without an in-depth exploration of the Borel group elements, such as done in \thm{CliffordExact}.  

For each integer $k\,{\in}\, \irange{0}{n}$ define 
\[
H^k := \hgate_1 \hgate_2 \cdots \hgate_k.
\]
Here $\hgate_i$ is the Hadamard gate acting on the $i$-th qubit ($H^0 \,{=}\, Id$). From Bruhat decomposition~\cite{maslov2018shorter} one infers that any $U\,{\in}\, \calC_n$ can be written (non-uniquely) as 
\be
\label{eq:Bruhat1}
U = L H^k R
\ee
for some $L,R \in \calF_n$ and integer $k\,{\in}\, \irange{0}{n}$. We would like to examine conditions under which the decomposition in \eq{Bruhat1} is unique. To this end, define a group
\[
\calF_n(k) = \calF_n \cap (H^k \calF_n H^k).
\]
$\calF_n(k)$ is indeed a group, since it is the intersection of 
two groups $\calF_n$ and $H^k \calF_n H^k$
(the latter is a group since $H^k$ is a self-inverse operator).

Decompose $\calF_n$ into a disjoint union of right cosets of $\calF_n(k)$, that is,
\[
\calF_n=\bigsqcup_{j=1}^{m} \calF_n(k) V_j.
\] 
Here we fixed some set of coset representatives $\{ V_j\in \calF_n\}$
and $m$ is the number of cosets (which depends on $n$ and $k$).
\begin{lemma}
	\label{lem:lem1}
	Any Clifford operator $U\,{\in}\, \calC_n$ can be uniquely written as
	\[
	U = L H^k V_j
	\]
	for some integers  $k\,{\in}\, \irange{0}{n}$, $j\,{\in}\, \irange{1}{m}$, and some operator $L\,{\in}\, \calF_n$.
	Accordingly, the Clifford group is a disjoint union of subsets $\calF_nH^k V_j$.
\end{lemma}

\begin{proof}
	We already know that $U=LH^kR$ for some $L,R\in \calF_n$, see \eq{Bruhat1}. Suppose
	\be
	\label{eq:unique1}
	U=L_1 H^{k_1} R_1 = L_2 H^{k_2} R_2
	\ee
	for some $L_i,R_i\in \calF_n$ and some integers $k_i\,{\in}\, \irange{0}{n}$.
	First we claim that $k_1{=}k_2$.  Indeed, 
	$|\la y|H^{k}|z\ra| \in \{0, 2^{-k/2}\}$ for any 
	basis vectors $y,z$ and any integer $k$. 
	Pick any basis vector $x$ such that
	$\la x|U|0^n\ra\ne 0$. 
	Since $L_i$ and $R_i$ map
	basis vectors to basis vectors (modulo phase factors), \eq{unique1} gives
	\[
	|\la x|U|0^n\ra| =2^{-k_1/2} = 2^{-k_2/2},
	\]
	that is, $k_1=k_2:= k$. From \eq{unique1} one gets $R_2 R_1^{-1} = H^k  L_2^{-1} L_1 H^k$ and thus $R_2R_1^{-1} \in H^k \calF_n H^k$. Since $R_2 R_1^{-1} \in \calF_n$, one infers that $R_2 R_1^{-1} \in \calF_n(k)$.  In particular, $R_1$ and $R_2$ belong to the same right coset of $\calF_n(k)$.  We conclude that the integer $k$ and the coset $\calF_n(k)R$ in the Bruhat decomposition \eq{Bruhat1} are uniquely determined by $U$.  Let the coset containing $R$ be $\calF_n(k)V_j$.
	
	The decomposition in \eq{Bruhat1} is invariant under the set of simultaneous transformations $R \mapsto A R$  and $L \mapsto L B$, where $A\,{\in}\, \calF_n(k)$ is arbitrary  and $B:= H^k A^{-1} H^k \in \calF_n(k)$.  Indeed, 
	\[
	(L B) H^k (A R) = L H^k A^{-1} H^k H^k A R = L H^k R.
	\]
	This transformation can be used to make $R=V_j$.
	Once the factors $R=V_j$ and $H^k$ are uniquely fixed,
	the remaining factor $L=U(H^k R)^{-1}$ is uniquely fixed. 
\end{proof}

We proceed to exploring the structure of $\calF_n(k)$.
\begin{lemma}
	Any element of $\calF_n(k)$ can be implemented by a quantum circuit shown in \fig{1}.
\end{lemma}

\begin{figure*}
\begin{center}
		\centerline{
			\Qcircuit @C=3em @R=1.2em{
				& & \ZZ_2^k & \ZZ_2^k & GL(k) & \ZZ_2^{k(n-k)} & \ZZ_2^{k(n-k)} & \\ 
				\lstick{k}		&{/} \qw	& \gate{\zgate}	& \gate{\notgate} 	& \gate{\cnotgate}& \ctrl{1} & \targ 	& \qw \\
				\lstick{n{-}k}	&{/} \qw	& \qw 			& \qw 				& \gate{\pgate{+}\notgate{+}\cnotgate}& \ctrl{-1} & \ctrl{-1} 	& \qw \\ 
				& &  &  & \calF_{n-k} &  &  & \\
			}
		}
		\caption{The subgroup $\calF_n(k)$ includes all $\hgate$-free operators $U$ such that conjugating $U$ by Hadamards on the first $k$ qubits gives another $\hgate$-free operator.  Illustrated are the circuit elements that can be used to construct an element of $\calF_n(k)$, as well as the groups they generate.}\label{fig:1}
\end{center}
\end{figure*}
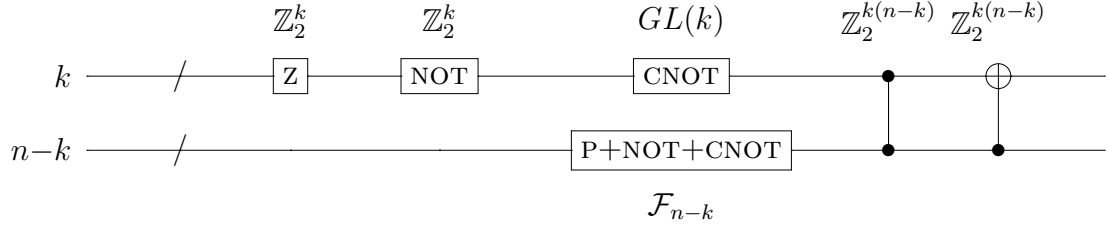

\begin{proof}
	Firstly, it is well known that 
	\[
	|\calC_n| = 2^{2n+n^2}\prod_{j=1}^n (4^j-1).
	\]
	Here, we count the elements of the Clifford group up to the global phase (of $\pm 1, \pm i$). Likewise, we count the number of elements in all other groups considered here up to the global phase. 
	
	Let $\text{Aff}(n)$ be the affine linear group over the binary field. It is isomorphic to the group generated by $\notgate$ and $\cnotgate$ gates, and isomorphic to the product $\ZZ_2^n \rtimes GL(n)$, being groups generated by the $\notgate$ gate and the $\cnotgate$ gates, correspondingly.
	We have
	\[
	|\text{Aff}(n)|=2^n \cdot \prod_{k=0}^{n-1}(2^n - 2^k) = 2^{n(n+1)/2} \prod_{j=1}^n (2^j-1).
	\]
	$\text{Aff}(n)$ is a subgroup of $\calF_n$, and factoring out affine linear transformations from $\calF_n$ leads to the leftover group of diagonal Clifford operations, that is equivalent to the combination of a stage of $\pgate$ gates and a stage of $\czgate$ gates \cite{maslov2018shorter}.  Note that the layer of $\pgate$ gates is isomorphic to $\ZZ_4^{n}$ (each qubit can have an $Id, \pgate, \pgate^2,$ or $\pgate^3=\pgate^\dagger$ be applied to it) and the layer of $\czgate$ is isomorphic to $\ZZ_2^{n(n-1)/2}$ (each of $n(n-1)/2$ $\czgate$ gates is either present or not).  Thus,  
	\[
	\calF_n \cong \ZZ_4^{n} \rtimes \ZZ_2^{n(n-1)/2} \rtimes \ZZ_2^n \rtimes GL(n),
	\]
	where the first factor represents group generated by $\pgate$ gates, the second factor is the group generated by $\czgate$ circuits, the third factor corresponds to the layer of $\notgate$ gates, and the fourth factor is the general linear group, equivalent to the $\cnotgate$ gate circuits. Thus
	\[
	|\calF_n|=2^{2n+n^2} \prod_{j=1}^n (2^j-1).
	\]
	Let $\calF_n'(k)$ be a subset of $\calF_n(k)$ implementable by the quantum circuit shown in \fig{1}. Specifically,
\begin{equation}\label{eq:fk} 
\calF_n'(k) \cong \ZZ_2^{k} \times  \ZZ_2^{k} \times GL(k) \times \calF_{n-k}\times \ZZ_2^{k(n-k)} \times \ZZ_2^{k(n-k)},
\end{equation}
where
\begin{enumerate}
		\item[1.] first $\ZZ_2^{k}$ corresponds to the group of unitaries implementable by Pauli-$\zgate$ gates applied to any of the top $k$ qubits;
		\item[2.] second $\ZZ_2^{k}$ corresponds to the group of unitaries implemented with Pauli-$\xgate$ gates applied to any of the top $k$ qubits;
		\item[3.] $GL(k)$ is the general linear group on first $k$ qubits, it is obtainable by the $\cnotgate$ gates;
		\item[4.] $\calF_{n-k}$ is the group of $\hgate$-free Clifford circuits spanning $n{-}k$ qubits;
		\item[5.] first $\ZZ_2^{k(n-k)}$ is the group of unitaries implementable as $\czgate$ gates with one input in the set of top $k$ qubits, and the other input in the set of bottom $n{-}k$ qubits;
		\item[6.] second $\ZZ_2^{k(n-k)}$ is the group of unitaries implementable as $\cnotgate$ gates with target in the set of top $k$ qubits, and control in the set of bottom $n{-}k$ qubits.
\end{enumerate}
	
Note that the conjugation by $H^k$ maps groups in items 1. and 2. as well as 5. and 6. into each other.  The conjugation by $H^k$ is an automorphism of the groups listed under items 3. and 4.  It can be shown by inspection of the individual gates that belong to the stages listed in \eq{fk} that all these subgroups are contained in $\calF_n(k)$. Simple algebra gives
\begin{equation}
\label{eq:lower2}
|\calF_n'(k)|= 2^{2n+n^2 - k(k+1)/2} \prod_{j=1}^{n-k} (2^j-1) \prod_{j=1}^k ( 2^j-1).
\end{equation}
By \lem{lem1} one has
\be
\label{eq:lower3}
|\calC_n| = |\calF_n| \cdot  \sum_{k=0}^n \frac{|\calF_n|}{ |\calF_n(k)|}
\le  |\calF_n| \cdot  \sum_{k=0}^n \frac{|\calF_n|}{|\calF_n'(k)|},
\ee
since $|\calF_n(k)|\ge |\calF_n'(k)|$.
We next show that 
\begin{equation}\label{eq:main}
|\calC_n| = |\calF_n| \cdot  \sum_{k=0}^n \frac{|\calF_n|}{|\calF_n'(k)|}.
\end{equation}
Combining \eq{lower3} and \eq{main} we conclude that 
$\calF_n(k)=\calF_n'(k)$ for all $k$ proving the lemma.

To prove \eq{main} it is convenient to use Gaussian binomial coefficients~\cite{WikipediaArticle} defined as follows,
\[
{n \choose k}_2 := \frac{ \prod_{j=1}^n (2^j-1)}{\prod_{j=1}^k (2^j-1) \prod_{j=1}^{n-k}(2^j-1)},
\]
where $0\le k\le n$. We need the following identity~\cite{WikipediaArticle}:
\[
\sum_{k=0}^n 2^{k(k-1)/2} {n \choose k}_2 t^k = \prod_{k=0}^{n-1} (1+t2^k).
\]
Here $t$ is a formal variable. Setting $t{=}2$ one can rewrite the above as
\be
\label{eq:gauss1}
\sum_{k=0}^n 2^{k(k+1)/2} {n \choose k}_2  = \prod_{k=1}^n (1+2^k).
\ee
Using the expression for $|\calF_n|$ and \eq{lower2} one gets
\[
\frac{|\calF_n|}{|\calF_n'(k)|}
= 2^{k(k+1)/2}{n \choose k}_2.
\]
The identity in \eq{gauss1} gives
\begin{eqnarray*}
|\calF_n| \cdot  \sum_{k=0}^n \frac{|\calF_n|}{|\calF_n'(k)|}
=|\calF_n|\cdot \prod_{k=1}^n (1+2^k) \\
=2^{2n+n^2} \prod_{k=1}^n (2^k+1)(2^k-1)
= 2^{2n+n^2} \prod_{k=1}^n (4^k-1)
=|\calC_n|,
\end{eqnarray*}
confirming \eq{main}. 
\end{proof}

\section{Python implementation of Algorithm~\ref{Sampling_QMallows}}
\label{app:C}

Here we provide a Python implementation of Algorithm~\ref{Sampling_QMallows}
in the main text.  Python language implementation is also \href{https://qiskit.org/documentation/_modules/qiskit/quantum_info/operators/symplectic/random.html#_sample_qmallows}{included in Qiskit} \cite{Qiskit} as a function $\_sample\_qmallows$.

\begin{lstlisting}[language=Python]
import numpy as np
def _sample_qmallows(n):

# Hadamard layer
h = np.zeros(n, dtype=int)

# Permutation layer
S = np.zeros(n, dtype=int)

A = list(range(n))

for i in range(n):
# number of elements in A
m = n - i 
r = np.random.uniform(0,1)
index = int(np.ceil(np.log2(1 + 
       (1 - r) * (4 ** (-m)))))
h[i] = 1*(index<m)
if index<m:
k = index
else:
k = 2*m - index -1
S[i] = A[k]
del A[k]
return h,S

\end{lstlisting}

\section{Python implementation of Algorithm~\ref{random_clifford_optimal}}
\label{app:D}

Here we provide a Python implementation of Algorithm~\ref{random_clifford_optimal} in the main text.  Python language implementation is also \href{https://qiskit.org/documentation/stubs/qiskit.quantum_info.random_clifford.html}{included in Qiskit} \cite{Qiskit} as a function $random\_clifford$.

\begin{lstlisting}[language=Python]
import numpy as np
def random_clifford(n):

assert(n<=200)

# constant matrices
ZR = np.zeros((n,n), dtype=int)
ZR2 = np.zeros((2*n,2*n), dtype=int)
I = np.identity(n, dtype=int)

h,S = _sample_qmallows(n)

Gamma1 = np.copy(ZR)
Delta1 = np.copy(I)
Gamma2 = np.copy(ZR)
Delta2 = np.copy(I)

for i in range(n):
Gamma2[i,i] = np.random.randint(2)
if h[i]:
Gamma1[i,i] = np.random.randint(2)

for j in range(n):
for i in range(j+1,n):
b = np.random.randint(2)
Gamma2[i,j] = b
Gamma2[j,i] = b
Delta2[i,j] = np.random.randint(2)
if h[i]==1 and h[j]==1:
b = np.random.randint(2)
Gamma1[i,j] = b
Gamma1[j,i] = b
if h[i]==1 and h[j]==0 and S[i]<S[j]: 
b = np.random.randint(2)
Gamma1[i,j] = b
Gamma1[j,i] = b
if h[i]==0 and h[j]==1 and S[i]>S[j]: 
b = np.random.randint(2)
Gamma1[i,j] = b
Gamma1[j,i] = b
if h[i]==0 and h[j]==1: 
Delta1[i,j] = np.random.randint(2)
if h[i]==1 and h[j]==1 and S[i]>S[j]:
Delta1[i,j] = np.random.randint(2)
if h[i]==0 and h[j]==0 and S[i]<S[j]:
Delta1[i,j] = np.random.randint(2)

# compute stabilizer tableaux 
PROD1 = np.matmul(Gamma1,Delta1)
PROD2 = np.matmul(Gamma2,Delta2)
INV1 = np.linalg.inv(np.transpose(Delta1))
INV2 = np.linalg.inv(np.transpose(Delta2))
F1 = np.block([[Delta1, ZR],[PROD1, INV1]])
F2 = np.block([[Delta2, ZR],[PROD2, INV2]])
F1 = F1.astype(int) % 2
F2 = F2.astype(int) % 2

# compute the full stabilizer tableaux
U = np.copy(ZR2)
# apply qubit permutation S to F2
for i in range(n):
U[i,:] = F2[S[i],:]
U[i+n,:] = F2[S[i]+n, :]
# apply layer of Hadamards
for i in range(n):
if h[i]==1:
U[(i,i+n),:] = U[(i+n,i),:]
# apply F1
return np.matmul(F1,U) % 2
\end{lstlisting}

\end{document}